\documentclass[journal]{IEEEtran}
\usepackage{cite}
\usepackage{amsmath,amssymb,amsfonts}
\usepackage{algorithmic}
\usepackage{graphicx}
\usepackage{algorithm,algorithmic}
\usepackage{hyperref}
\usepackage{subcaption}
\hypersetup{hidelinks} 
\usepackage{textcomp}
\usepackage{bm}

\usepackage{fvextra}

\usepackage{cleveref}
\crefformat{equation}{Eq.\,(#2#1#3)}
\crefname{section}{Sec.}{Secs.}
\crefformat{section}{Sec.\,#2#1#3}
\Crefformat{section}{Section\,#2#1#3}
\crefname{figure}{Fig.}{Figs.}
\crefformat{figure}{Fig.\,#2#1#3}
\Crefformat{figure}{Figure\,#2#1#3}

\usepackage{dsfont}
\usepackage[T1]{fontenc}

\newcommand{\secref}[1]{Sec.\,\ref{#1}}

\newcommand{\ie}{\emph{i.e.}}

\usepackage{amsthm}
\newtheorem{theorem}{Theorem}
\newtheorem{proposition}{Proposition}
\newtheorem{corollary}{Corollary}

\theoremstyle{definition}

\newtheorem{remark}{Remark}
\newtheorem{assumption}{Assumption}

\newcommand{\mI}{{\mathds{1}}}
\newcommand{\R}{\mathbb{R}}
\newcommand{\mZ}{\mathbb{Z}}

\newcommand{\Xset}{\mathbb{X}}
\newcommand{\Uset}{\mathbb{U}}

\begin{document}
\title{Physics-informed Reinforcement Learning for Stochastic Reach-Avoid Analysis}

\author{
Hikaru Hoshino, 
Yorie Nakahira
\vspace*{-5mm}
}

\maketitle

\begin{abstract}
Stochastic reach-avoid analysis of controlled dynamical systems is an important tool for safety-critical control under uncertainty, in which the reach-avoid probability is characterized by a Hamilton-Jacobi partial differential equation (PDE). However, solving this PDE using conventional numerical methods becomes computationally intractable as the system dimension increases. Physics-informed neural networks (PINNs) may converge to inaccurate local minima when trained primarily through PDE-residual minimization. Reinforcement learning (RL) offers a scalable alternative, but its learned value functions may be inaccurate or inconsistent with the governing PDE. This paper proposes a physics-informed RL (PIRL) framework that combines the complementary strengths of PINNs and RL for stochastic reach-avoid analysis. We develop a scheduled PIRL algorithm in which temporal-difference actor–critic learning first guides the critic toward a meaningful approximation of the reach-avoid value function. PDE-residual and boundary-condition losses are then introduced progressively to enforce consistency with the governing PDE and its boundary conditions. The proposed method mitigates the failure modes of conventional PINN techniques while achieving accuracy comparable to that of successfully trained PINNs. The effectiveness of the proposed framework is demonstrated through two case studies. 
\end{abstract}


\section{Introduction} \label{sec:introduction}

\IEEEPARstart{S}{tochastic} 
reach-avoid analysis is crucial for evaluating safety and synthesizing controllers under uncertainty~\cite{Abate2008,Summers2010,MohajerinEsfahani2016}. 
Given a target set and an avoid set, the objective is to characterize the set of initial states from which a controlled stochastic system can reach the target before entering the avoid set with probability greater than a prescribed threshold. A special case of these problems reduces to evaluating the reach-avoid probability, which can be computed by solving a Hamilton–Jacobi–Bellman (HJB)-type partial differential equation (PDE)~\cite{MohajerinEsfahani2016}. 
This formulation naturally captures both reachability and safety of controlled dynamical systems, and is particularly effective for complex dynamical systems operating in unpredictable environments, where noise and uncertainties pose significant concerns in operation and safety. 

A fundamental challenge in stochastic reach-avoid analysis is to accurately approximate the reach-avoid probability, particularly under maximally safe policies.
Solving the governing PDE using conventional numerical methods becomes computationally intractable as the system dimension increases.
Physics-informed neural networks (PINNs)~\cite{Raissi2019,Luo2025} provide a promising alternative by approximating PDE solutions with neural networks trained through PDE-residual and boundary-condition losses.
However, purely residual-based PINN training may be unreliable because the PDE residual can be reduced even when the learned solution converges to a poor local minimum or a nearly trivial solution~\cite{Leiteritz2021, S.Wang2023:PinnGuide}.
Reinforcement learning (RL), on the other hand, learns policies and value functions directly from sampled trajectories and has achieved remarkable success in high-dimensional optimal control and reachability problems~\cite{Akametalu2018,Fisac2019,Hsu2021,Ganai2024}.
However, standard RL provides only temporal-difference (TD) supervision along sampled trajectories, making it difficult to ensure the PDE consistency and accuracy of the learned value function.
Motivated by these complementary strengths and limitations, this paper proposes a physics-informed reinforcement learning (PIRL) framework that combines trajectory-based actor-critic learning with PDE-based constraints.
The proposed approach uses RL to guide the critic toward a meaningful reach-avoid value function and thereby mitigate the failure modes of residual-based PINN training.
The PDE residual and boundary-condition losses are then progressively incorporated to improve the PDE consistency and accuracy of the learned value function.

\subsection{Related work}

\subsubsection{Reachability Analysis Methods}

Classical reachability analysis methods require discretization of the state space, and their computational cost grows exponentially with the number of state variables.
As a result, their practical application has largely been limited to low-dimensional systems, typically up to four or five state dimensions~\cite{Mitchell2005}, unless special problem structures are exploited~\cite{Bansal2017}.
This scalability limitation has motivated the use of deep RL. 
Early work introduced RL formulations of deterministic Hamilton-Jacobi reachability problems that enabled Q-learning and TD methods~\cite{Akametalu2018,Fisac2019,Hsu2021}.
Building on this direction, RL has been used to approximate reachability-related value functions, enabling reachability analysis to be applied to higher-dimensional control problems as reviewed in \cite{Ganai2024}.
However, these deterministic methods consider bounded disturbances or adversarial uncertainty, rather than stochastic reach-avoid probabilities governed by diffusion processes. 

For stochastic systems, several methods have been developed to provide probabilistic safety and reach-avoid guarantees. Dynamic-programming and linear-programming formulations can directly characterize stochastic reachability and reach-avoid probabilities with rigorous guarantees~\cite{Abate2008,Summers2010,Schmid2023}. 
However, these approaches rely on gridding or finite basis-function approximations and generally suffer from poor scalability in high-dimensional state spaces. 
Sampling-based approaches can estimate pointwise stochastic reachability probabilities in high-dimensional systems~\cite{Thorpe2020,Thorpe2022}. 
However, they do not directly characterize the reach-avoid value function over the state space. 
For stochastic reach-avoid analysis or safe control, underapproximation and certificate-based methods have been proposed, including Lagrangian approximations~\cite{Gleason2017}, sum-of-squares and semidefinite-programming methods~\cite{Xue2021}, and barrier-like conditions that certify lower bounds on reach-avoid probabilities~\cite{Xue2026}.
Learning-based methods based on reach-avoid supermartingales have also been proposed~\cite{Zikelic2023}.
Although these approaches provide valuable probabilistic guarantees, they typically yield lower-bound certificates or conservative inner approximations of probability-threshold sets, rather than directly approximating the maximal reach-avoid probability over the state-time domain.

\subsubsection{Physics-informed Reinforcement Learning (PIRL)} 
\label{sec:review_pirl}

The term PIRL has been used broadly to refer to RL methods that incorporate physical knowledge into the learning process. 
As reviewed in~\cite{Banerjee2025}, existing PIRL methods exploit such knowledge in various forms, including state and action representations, reward design, policy regularization, safety constraints, and model learning. Most of these approaches use physics to regularize policy learning or improve model prediction, rather than to exploit the PDE characterization of the optimal value function in RL.

A more closely related direction is PIRL methods based on PINNs~\cite{Raissi2019}, which incorporate PDE constraints into neural-network training by penalizing differential-equation residuals and boundary-condition errors.
In the optimal control literature, PINNs have been used to approximate value functions and solve HJB-type equations. 
PINN-based policy iteration methods have been developed for optimal regulator problems~\cite{Meng2024,Wang2024PIRL}, while PDE-based value-function losses have also been incorporated into model-free RL in~\cite{Mukherjee2023}.
These approaches demonstrate the benefit of combining value-function learning with PDE constraints, but do not specifically address the optimization failure modes of residual-based training.
PINN training is known to suffer from such difficulties, including poor local minima and trivial solutions~\cite{Leiteritz2021,S.Wang2023:PinnGuide}.
This issue is particularly relevant to reach-avoid problems, where value information must be propagated from the boundary conditions rather than from interior rewards. 
In such problems, residual-based PINN training may converge to an incorrect or nearly trivial value function despite achieving a small PDE residual, as illustrated later in the numerical experiments (\cref{fig:1D_performance}).


\subsection{Contributions}

Motivated by these limitations, this paper proposes a PIRL framework for stochastic reach-avoid analysis of controlled dynamical systems that combines trajectory-based RL with physics-informed value-function learning. The main methodological contributions are summarized as follows.

First, we establish an RL formulation for stochastic reach-avoid analysis that directly represents the reach-avoid probability under maximally safe policies.
The reach-avoid objective is naturally expressed as a max-multiplicative cost function~\cite{Abate2008,Summers2010}, which is not directly compatible with standard RL algorithms based on additive cumulative rewards.
We show that, by augmenting the state with the remaining horizon and introducing absorbing target and avoid sets, this max-multiplicative reach-avoid probability can be transformed into an equivalent additive-reward RL problem (Proposition\,\ref{prop:RL}).
Unlike RL formulations developed for deterministic reach-avoid problems~\cite{Akametalu2018,Fisac2019,Hsu2021}, the proposed formulation applies to stochastic reach-avoid probabilities and enables standard actor-critic methods to learn both the maximal probability and its optimizing policy. 
Moreover, rather than providing only a lower bound or a conservative inner approximation of a probability-threshold set, as in certificate-based approaches~\cite{Gleason2017,Xue2021,Zikelic2023,Xue2026}, it directly approximates the underlying maximal reach-avoid probability.
We further establish the continuous-time characterization of the resulting RL value function, showing that it converges to the stochastic reach-avoid value function characterized by the corresponding HJB equation (Theorem\,\ref{thm:pde_characterization}). 
The residual-based PINN error analysis of \cite{Wang2026} is then applied to the policy-evaluation PDE associated with the learned policy, which yields an approximation-error bound in terms of the PDE residual and boundary-condition mismatch (Corollary 1). This bound provides a theoretical interpretation of the physics-informed losses used in the proposed framework.

Second, we develop a PIRL algorithm that combines TD learning from sampled trajectories with PDE residual and boundary-condition losses. 
While related PIRL approaches incorporate PINN-based losses into value-function learning~\cite{Meng2024,Wang2024PIRL,Mukherjee2023}, our approach specifically addresses the failure modes of residual-based PINN training through a scheduled transition between trajectory-based and physics-informed learning.
The proposed algorithm emphasizes TD learning in the initial stage to guide policy optimization and steer the critic toward a meaningful reach-avoid solution, and then gradually introduces the PDE residual and boundary-condition losses during continuation training.
This scheduling prevents the physics-informed losses from dominating before a meaningful solution has been learned, thereby mitigating the incorrect or nearly trivial solutions associated with residual-based physics-informed value-function learning (\cref{fig:1D_performance}). 
We also introduce a large language model (LLM)-guided outer-loop procedure for automatically adjusting the continuation schedule based on the outcomes of completed training runs (\cref{fig:drift_training_loss}).

The proposed PIRL is evaluated in both a one-dimensional benchmark and a nonlinear vehicle drifting problem~\cite{Hindiyeh2014}. 
The  one-dimensional example shows that the scheduled PIRL method avoids the failure modes of PINN-based policy iteration and accurately recovers the reach-avoid value function. 
The drifting example demonstrates applicability to an eight-dimensional time-augmented state space, beyond the dimensional range typically addressed by grid-based PDE solvers without special structure. 
The learned value function captures the recoverable region around the unstable drift equilibrium and is consistent with both projected phase-space structures and closed-loop stochastic rollout outcomes (\cref{fig:drift_value_contours,fig:drift_results}). 
These results demonstrate that the proposed method can recover meaningful reach-avoid structures in a challenging nonlinear control problem that cannot be fully characterized through individual low-dimensional projections.

Preliminary versions of this work partly appeared in \cite{HoshinoACC2024} and \cite{HoshinoITSC2024}.
Compared with \cite{HoshinoACC2024}, the present paper extends the framework from invariance-based maximal safety probability estimation to stochastic reach-avoid analysis of controlled dynamical systems and newly demonstrates residual-based PINN training can suffer from nearly trivial solutions. 
Compared with \cite{HoshinoITSC2024}, which used Deep Q-Network (DQN)~\cite{Mnih15}-based PIRL algorithm for vehicle drifting, the present paper develops a Twin-Delayed Deep Deterministic Policy Gradient (TD3)~\cite{Fujimoto2018:TD3}-based PIRL algorithm with a new loss scheduling strategy, and validates the resulting physics-informed value function in both
benchmark and drifting examples.

\subsection{Notation}

Let $\mathbb{R}$ and $\mathbb{R}_{+}$ be the sets of real and nonnegative real numbers, respectively, and let $\mathbb{Z}_{+}$ be the set of nonnegative integers.
For a set $A$, $A^{c}$ and $\partial A$ are its complement
and boundary, respectively.
The indicator of a set $A$ is denoted by $\mathbf{1}_{A}(x)$, which equals $1$ if $x\in A$ and $0$ otherwise.
The Euclidean norm is denoted by $\|\cdot\|$, and $\lfloor a\rfloor$ represents the greatest integer less than or equal to $a\in\mathbb{R}$.
For a random variable $Z$, $\mathbb{E}[Z]$ is its
expectation, and $\mathbb{E}[Z\mid Y=y]$ is the conditional expectation given $Y=y$.
Similarly, $\mathbb{P}(E\mid Y=y)$ denotes the conditional
probability of an event $E$ given $Y=y$.
Uppercase letters denote random variables or stochastic
processes, whereas the corresponding lowercase letters denote
their realizations.
For a sufficiently smooth scalar function $\varphi$,
$\partial_t\varphi$, $\partial_x\varphi$, and
$\partial_x^2\varphi$ are its time derivative, gradient
with respect to $x$, and Hessian with respect to $x$,
respectively.
For a matrix $M$, $\operatorname{Tr}(M)$ represents its trace.
The spaces of $k$-times continuously differentiable functions
and continuous functions are denoted by $C^k$ and $C^0$,
respectively, while $C^{1,2}$ denotes the space of functions
that are once continuously differentiable in time and twice
continuously differentiable in the state variables.

\section{Problem Statement}

This section introduces the foundations of our problem setting.
The stochastic system dynamics considered in this paper are defined in \cref{sec:dynamics}. 
The stochastic reach-avoid problem is introduced in \cref{sec:reachability}. 
The PDE characterization of the value function is presented in \cref{sec:pde}.

\subsection{System Dynamics} \label{sec:dynamics}

Let $(\Omega, \mathcal{F}, \mathbb{F}, \mathbb{P})$ be a filtered probability space where the filtration $\mathbb{F} = \{\mathcal{F}_t\}_{t\ge0}$ satisfies the usual conditions \cite{Karatzas1998}, \ie,  $\mathcal{F}_t$ is right-continuous and $\mathcal{F}_0$ contains all the $\mathbb{P}$-negligible events in $\mathcal{F}$, and let $\{W_t\}_{t\ge0}$ be a {\it w}-dimensional $\{\mathcal{F}_t\}$-standard Brownian motion with $W_0=0$.
We consider a stochastic dynamical system defined over an open domain $\Xset \subset \mathbb{R}^n$, whose state process $\{X_t\}_{t \ge 0}$ evolves according to the following stochastic differential equation (SDE):
\begin{align}
\mathrm{d}X_t = f(X_t, U_t) \mathrm{d}t + \sigma(X_t, U_t) \mathrm{d}W_t, \label{eq:sde}
\end{align}
where $f: \R^n \times \R^m \to \mathbb{R}^n$ and $\sigma: \R^n \times \mathbb{R}^m \to \mathbb{R}^{n \times w}$. 
The control process $\bm{u} = \{U_t\}_{t \ge 0}$ takes values in a compact set $\mathbb{U} \subset \mathbb{R}^m$, and is assumed to be an $\mathbb{F}$-progressively measurable process.
We denote by $\mathcal{U}$ the set of all such admissible control processes.
Throughout this paper, we assume that the functions $f$ and $\sigma$ are continuous and satisfy a uniform Lipschitz condition with respect to the first argument (state), uniformly over the second argument (control input). 
These regularity conditions ensure that the SDE \eqref{eq:sde} admits a unique strong solution~\cite[Sec.\,IV.2]{Fleming06}.
We denote this strong solution by $\{X_t^{(x,\bm{u})}\}_{t \ge 0}$ to explicitly indicate its dependence on the initial state $x$, and the control process $\bm{u}$. 


\subsection{Stochastic Reach-Avoid Problem} \label{sec:reachability}

Safety analysis of dynamical systems typically relies on two notions: \emph{reachability}, which concerns the property of a system reaching desirable states, and \emph{invariance}, which ensures that the system avoids unsafe regions. 
These two properties can be seen as dual aspects of safety, and their relationship has been extensively studied in both deterministic settings~\cite{Lygeros2004,Liao2022} and stochastic settings~\cite{Schmid2023,MohajerinEsfahani2016}.
A unified formulation that captures both aspects is the so-called \emph{reach-avoid} problem, and our treatment follows the framework developed in~\cite{MohajerinEsfahani2016}, where stochastic reach-avoid problems are formulated in terms of exit-time problems with discontinuous payoff functions.

Let $A, B \subset \Xset$ be disjoint sets, where $A$ is the target set and $B$ is the avoid set. 
For 
a fixed time horizon $T > 0$, the reach-avoid set with a threshold $p \in [0,1]$ is defined as the set of initial states from which the system can reach $A$ before entering $B$ within $[0,T]$ with probability greater than $p$:
\begin{align}
 \mathrm{RA}(A,B; p) := & 
 \Bigl\{  x \in \Xset \,\Big|\, \exists \bm{u} \in \mathcal{U} \text{ such that } 
   \notag \\ & \hspace{3mm}  
 \mathbb{P}\Bigl[\exists s \in [0, T]: X_s^{(x,\bm{u})} \in A \text{ and } 
 \notag \\ & \hspace{3mm} 
 X_r^{(x,\bm{u})} \notin B, \forall r \in [0,s] \Bigr] > p \Bigr\}.
\end{align}
The probability described above can be equivalently characterized using the exit-time formulation. In \cite{MohajerinEsfahani2016}, this problem is studied by formulating a family of subproblems starting from arbitrary initial time $t\in[0,T]$. 
For each $t$, they construct a corresponding subfiltration $\mathbb{F}_t$ and define a SDE whose solution evolves from initial state $X_t=x$ over the time interval $[t,T]$.
Following this approach, we denote by $\{ X_{s,t}^{(x,\bm{u})} \}_{s \ge t}$ the strong solution to the SDE starting at time $t$ from state $x$, under control process $\bm{u}\in  \mathcal{U}_t$, where $\mathcal{U}_t$ stands for the set of $\mathbb{F}_t$-progressively measurable maps into $\mathbb{U}$. 
Using this notation, the reach-avoid value function is defined as 
\begin{align}
V(t, x) := \sup_{\bm{u} \in \mathcal{U}_t} \mathbb{E}\left[ \mathds{1}_A \left(X_{t_\mathrm{exit},t}^{(x, \bm{u})} \right)~\middle|~ X_t=x, \bm{u} \right] \label{eq:exit_value}
\end{align}
with $t_\mathrm{exit} := \min(\tau_{A \cup B},  T)$, where $\tau_{A \cup B}$ stands for the first entry time to $A \cup B$. 
It is shown in \cite[Proposition~3.3]{MohajerinEsfahani2016} that if $A$ and $B$ are disjoint closed sets, the reach-avoid set can be expressed as the superlevel set of the value function:
\begin{align}
\mathrm{RA}(A, B; p) = \left\{ x \in \Xset \,\middle|\, V(0, x) > p \right\}. \label{eq:ra_levelset}
\end{align}
This formulation establishes a direct connection between the reach-avoid problem and stochastic optimal control. In particular, the function $V(t,x)$ can be interpreted as the maximal probability that the system reaches the target $A$ before entering the avoid set $B$ within the horizon $[t,T]$.

\begin{remark}
  Note that invariance properties requiring the state to remain in a safe set $C$ over the entire time horizon $[0,T]$ can also be expressed within the same reach-avoid structure by setting $A := C$ and $B := \Xset \setminus C$, and additionally requiring that the reach condition be satisfied at the terminal time $T$ only. This formulation allows invariance problems to be analyzed using the same value-function characterization and computational framework developed in the following sections.
  \label{remark:invariance}
\end{remark}

\subsection{PDE Characterization} \label{sec:pde}

The value function defined in \cref{eq:exit_value} plays a central role in characterizing the stochastic reach-avoid set. 
To analyze its properties and develop computational methods, we consider its characterization as the solution of a HJB partial differential equation. 
To proceed, we introduce a set of technical assumptions under which the value function is characterized as a viscosity solution to the HJB equation. 
\begin{assumption} \label{assumption:HJB}
 We stipulate that 
 \begin{enumerate}
   \item[(a)] The domain $X \subset \mathbb{R}^n$ is a bounded open set, and the set $O := \mathbb{X} \setminus (A \cup B)$ has a boundary $\partial O$ that is a manifold of class $C^3$.
   \item[(b)] The controlled process is uniformly non-degenerate on $\overline{X}$, \ie, there exists $\delta>0$ such that for all $(x,u) \in\Xset \times \Uset$ and  $\xi \in \R^n$, $\xi^\top \sigma(x,u)\sigma^\top(x,u) \xi \ge \delta |\xi|^2$. 
 \end{enumerate}
\end{assumption}
Under Assumption~\ref{assumption:HJB}, it follows from \cite[Theorem~4.7]{MohajerinEsfahani2016} that the function $V(t,x)$ is a (discontinuous) viscosity solution to the following HJB equation:
\begin{align}
 - \partial_t V(t,x) - \sup_{u \in \mathbb{U}} \mathcal{H}^u V(t,x) = 0, \quad \text{on } [0,T) \times O \label{eq:HJB}
\end{align}
with boundary and terminal conditions
\begin{align}
& V(t,x) = \mathds{1}_A(x), \notag \\ & \hspace{5mm} \text{for } (t,x) \in [0,T) \times \partial O \cup \{T\}\times \overline{O}, \label{eq:HJB_BC}
\end{align}
where the operator $\mathcal{H}^u$ is given by
\begin{align}
\mathcal{H}^u \phi(t,x) := \,& f(x,u)^\top \partial_x \phi(t,x) 
 \notag \\ & 
 + \frac{1}{2} \operatorname{Tr} \left[ \sigma(x,u)\sigma(x,u)^\top \partial_x^2 \phi(t,x) \right]. \label{eq:generator}
\end{align}

\begin{remark}
In the case of an invariance problem, the HJB equation~\eqref{eq:HJB} is recovered by choosing the domain $O=\Xset\setminus B=C$, with the boundary conditions
\begin{align}
    V(t,x) &= 0, && (t,x)\in [0,T)\times\partial C,\\
    V(T,x) &= 1, && x\in C,
\end{align}
which correspond to requiring that the system remain in the safe set throughout the time horizon $T$.  
\end{remark}

Although the viscosity solution framework provides a complete characterization of the reach-avoid value function, the discontinuous boundary condition~\eqref{eq:HJB_BC} prevents the associated HJB problem from admitting a classical solution. 
To obtain a smooth HJB problem suitable for the subsequent physics-informed learning framework, we adopt the smoothed reach-avoid formulation proposed in~\cite{MohajerinEsfahani2016}. 
Specifically, the target set $A$ is replaced by a smaller set $A_\epsilon\subset A$, defined as
\begin{align}
A_\epsilon := \left\{ x \in A \,\middle|\, \mathrm{dist}(x, A^\mathrm{c}) \ge \epsilon \right\},
\end{align}
where $\mathrm{dist}(x, A) := \inf_{y \in A} \| x - y \| $. 
This modification introduces a margin around the boundary of $A$, thereby avoiding discontinuities in the terminal condition of the value function.
For example, one can define a continuous approximation of the indicator function $\mathds{1}_A$ as follows:
\begin{align}
    l_\epsilon(x) = \max \left\{ 1 - \frac{\mathrm{dist}(x, A_\epsilon)}{\epsilon}, \, 0 \right\}, \label{eq:continuous_terminal_cost}
\end{align}
which is Lipschitz continuous and takes values in \([0,1]\). It is shown in \cite[Theorem~5.1]{MohajerinEsfahani2016} that the modified value function
\begin{align}
    V_\epsilon(t, x) := \sup_{\bm{u} \in \mathcal{U}_t} \mathbb{E}\left[ l_\epsilon \left(X_{t_\mathrm{exit},t}^{(x,\bm{u})} \right)~\middle|~ X_t=x, \bm{u} \right], 
\end{align}
with the function $ l_\epsilon(x) $ chosen as in \eqref{eq:continuous_terminal_cost} is continuous and $V(t,x) = \lim_{\epsilon\to 0} V_\epsilon(t,x)$. 
To further ensure classical solvability, we henceforth consider a sufficiently smooth approximation $l_\epsilon$ satisfying the following assumptions:
\begin{assumption} \label{assumption:smoothness}
 We stipulate that 
 \begin{enumerate}
 \item[(a$'$)]
The domain $X \subset \mathbb{R}^n$ is a bounded open set, and the set $O_\epsilon$ has a boundary $\partial O_\epsilon$ that is a manifold of class $C^3$.
   \item[(c)]  The smoothed indicator function $l_\epsilon$ is chosen such that $\mathds{1}_A(x) = \lim_{\epsilon \to 0} l_\epsilon(x)$ and  $l_\epsilon \in C^3(\overline{\Xset})$.
   Furthermore, $l_\epsilon$ satisfies necessary parabolic compatibility conditions of order 1 on the corner manifold $\{T\} \times \partial O_\epsilon$.
   \item[(d)] The functions $f$ and $a:=\sigma\sigma^\top$, together with their first and second partial derivatives with respect to the state, are continuous on $\overline X\times U$.
 \end{enumerate}
\end{assumption}
Under Assumptions~\ref{assumption:HJB} and \ref{assumption:smoothness}, it follows from standard results in stochastic control theory (e.g., \cite[Sec.\,IV.4, Theorem 4.1]{Fleming06}) that the value function $V_\epsilon$ is the unique classical solution to the associated HJB equation.
Furthermore, in this classical setting, a measurable maximizing selector
$\mu:[0,T)\times X\to U$ exists such that
\begin{align}
  \mu(t,x)
  \in \arg\max_{u\in U} H^u V_\epsilon(t,x).
\end{align}
Under the uniformly parabolic assumptions, the corresponding optimal
Markov control can be constructed as in~\cite[Sec.~IV.4]{Fleming06}.
We henceforth restrict our attention to measurable Markov policies and
denote the set of such policies by $\mathcal U_M$.

\section{RL Characterization of Reach-Avoid Problem} \label{sec:model_free_RL}

While the HJB formulation presented in the previous section provides a complete characterization of the reach-avoid value function, directly solving the HJB equation remains computationally challenging for high-dimensional systems. 
RL offers a scalable alternative by approximating value functions from sampled trajectories. 
The main result of this section is that the reach-avoid probability admits an exact RL characterization:
by reformulating the reach-avoid objective as an equivalent additive reward, the reach-avoid value function is shown to coincide exactly with the RL value function.
This equivalence provides the theoretical foundation of the proposed physics-informed RL framework.

We begin by introducing a discrete-time formulation of the original non-smooth reach-avoid problem with a constant step size $\Delta t$ under piecewise constant control processes.
For $0=t_0 < t_1 < \dots < t_k < \dots$, where $t_k := k\Delta t$, $k \in \mZ_+$, by defining the discrete-time state $X_k := X_{t_k}$ with an abuse of notation, the discretized system can be written as
\begin{align}
  X_{k+1} = F^\mu(X_k, \Delta W_k), \label{eq:descrite_system}
\end{align}
where $\Delta W_k := \{ W_t \}_{t \in [t_k, t_{k+1})}$ is a disturbance signal, and $F^{\mu}$ stands for the state transition map derived from the SDE \eqref{eq:sde} under a Markov control policy $\mu$. 
Note that using a piece-wise constant control process with a Markov policy $\mu$ implies that the control process is given as $U_t = \mu(\delta(t), X_{\delta(t)})$, for $t\in \R_+$,  where $\delta(t) := \lfloor t/\Delta t \rfloor \Delta t$, and the discretized system \eqref{eq:descrite_system} has the Markov property at the discrete times as discussed in \cite{Mao2013}. 
Then, for an arbitrarily chosen horizon length $\tau \in \R_+$, consider the reach-avoid problem where the goal is to reach the target set $A$ within $\mathcal{N}_\tau := \{0,\dots, N(\tau)\}$, where $N(\tau) := \lfloor \tau/\Delta t \rfloor$, before entering the avoid set $B$. 
Under the discretized dynamics \eqref{eq:descrite_system} and a given control policy $\mu$, the associated reach-avoid probability is defined as
\begin{align} 
 \Psi^\mu_\mathrm{RA}(\tau, x) := \,&  
    \mathbb{P} \Bigl[ \exists k \in \mathcal{N}_\tau \text{ such that }  X_k \in A \,\text{ and } 
\notag \\ & \hspace{4mm}
   X_j \notin B,\, \forall j \in \{0,\dots, k\} \,\Big|\, X_0 = x, \mu \Bigr]. \label{eq:reach_avoid_prob}
\end{align}
The objective is to compute the maximal reach-avoid probability, defined as
\begin{align} \label{eq:descrete_RA_prob}
 \Psi_\mathrm{RA}^\ast(\tau,x) := \sup_{ \mu \in \mathcal{U}_\mathrm{M} } \Psi^\mu_\mathrm{RA}(\tau,x).
\end{align}

\begin{remark}
The objective in~\eqref{eq:descrete_RA_prob} can be naturally viewed as a stochastic optimal control problem with a max-multiplicative cost structure, since the cost can be rewritten as 
\begin{align}
  \Psi^\mu_\mathrm{RA}(\tau, x)
  =
  \mathbb{E}
  \left[
      \max_{k\in\mathcal N_\tau}
          \prod_{j=0}^{k}
          \mathds{1}_{B^\mathrm c}(X_j)
      \mathds{1}_A(X_k)
      \,\middle|\,
      X_0=x,\mu
  \right].
  \label{eq:multiplicative_cost}
\end{align}
Such multiplicative formulations have been studied in~\cite{Abate2008,Summers2010}, where Bellman equations (or  dynamic programming equations) for max-multiplicative cost-to-go functions are derived for stochastic reach-avoid analysis. 
This Bellman formulation has also served as the basis for verification methods that construct conservative lower bounds of the reach-avoid probability (e.g.,~\cite{Gleason2017,Xue2021,Xue2026,Zikelic2023}).
However, these multiplicative formulations are not directly compatible with standard RL algorithms, which assume additive cumulative rewards and form the basis of temporal-difference learning.
\end{remark}

The following construction transforms the problem \eqref{eq:descrete_RA_prob} into an equivalent additive reward representation.
The key idea is to augment the state with the remaining time until the horizon is reached and to make the target and avoid sets absorbing.
Specifically, we first introduce an auxiliary variable $T_k$ that represents the remaining outlook horizon, \emph{i.e.}, 
\begin{align}
    T_0 = \tau, \quad T_{k+1} = T_{k} - \Delta t, \quad \forall k \in \mZ_+.
\end{align}
Next, consider the auxiliary dynamics
\begin{align}
    \tilde X_{k+1}
    =
    \tilde F^\mu(\tilde X_k,\Delta W_k),
\end{align}
where
\begin{align}
\tilde F^\mu(\tilde X_k,\Delta W_k)
:=
\begin{cases}
F^\mu(\tilde X_k,\Delta W_k),
&
\tilde X_k\notin A\cup B,
\\
\tilde X_k,
&
\tilde X_k\in A\cup B.
\end{cases}
\label{eq:absorbing_structure}
\end{align}
Thus, the sets $A$ and $B$ become absorbing.
The augmented state space is defined by
$\mathcal{Y} := \R \times \mathbb{X} \subset \R^{n+1}$
with augmented state
\begin{align}
Y_k=[T_k,\tilde X_k^\top]^\top.
\end{align}
The corresponding augmented dynamics are
\begin{align}
Y_{k+1}
=
G^\mu(Y_k,\Delta W_k),
 \label{eq:augmented_dynamics}
\end{align}
where
\begin{align}
G^\mu(Y_k,\Delta W_k)
=
\begin{bmatrix}
T_k-\Delta t\\
\tilde F^\mu(\tilde X_k,\Delta W_k)
\end{bmatrix}.
\end{align}
The following proposition shows that the resulting additive reward formulation is exactly equivalent to the original reach-avoid problem.

\begin{proposition} \label{prop:RL}
Consider the system \eqref{eq:augmented_dynamics} starting from an initial state $y = [\tau, x^\top]^\top \in \mathcal{Y}$ and the reward function $r: \mathcal{Y} \to \R$ given by
\begin{align} \label{eq:reward}
  r(Y_k) := \mI_{\mathcal{T}}(T_k)\, \mI_A( \tilde{X}_k ) 
\end{align}
with $\mathcal{T} := [0, \Delta t)$. 
Then, for a given control policy $\mu \in \mathcal{U}_\mathrm{M}$, the RL value function $v^\mu_\mathrm{RL}$ defined by  
\begin{align}
  &v^\mu_\mathrm{RL}(y) := \mathbb{E} \left[ \sum_{k=0}^\infty r( Y_k ) ~\middle | ~ Y_0= y, \mu  \right]  \label{eq:rl_value_function}
\end{align}   
takes a value in $[0,1]$ and is equivalent to the reach-avoid probability $\Psi^\mu_\mathrm{RA}$, i.e., 
\begin{align}
 v^\mu_\mathrm{RL}(y) = \Psi^\mu_\mathrm{RA}(\tau, x).     
\end{align}
\end{proposition}

\begin{proof}
    See Appendix~\ref{proof:RL}. 
\end{proof}

\begin{remark}
The proposed RL formulation can also be applied to stochastic invariance problems.
As discussed in Remark~\ref{remark:invariance}, invariance can be formulated as a reach-avoid problem by setting $A:=C$, $B:=\mathbb X\setminus C$, and requiring that the target set be reached only at the terminal time.
In the augmented dynamics, this is achieved by modifying the absorbing dynamics so that only the avoid set $B$ is treated as absorbing, while trajectories reaching the safe set $A$ continue to evolve until the remaining time reaches zero.
The resulting additive reward representation follows from the same argument as in Proposition~\ref{prop:RL}.
A detailed proof is available in our previous conference publication~\cite{HoshinoACC2024}.
\end{remark}

\begin{remark}
Proposition~\ref{prop:RL} allows the reach-avoid problem to be solved using standard RL algorithms such as Deep Q-Network (DQN)~\cite{Mnih15}, Deep Deterministic Policy Gradient (DDPG)~\cite{Lillicrap2015:DDPG}, and Twin-Delayed Deep Deterministic Policy Gradient (TD3)~\cite{Fujimoto2018:TD3}.
In practical implementations, the absorbing states corresponding to the target and avoid sets are naturally treated as terminal states of an episode.
Moreover, since the cumulative reward in~\eqref{eq:rl_value_function} takes values only in $[0,1]$ and at most one nonzero reward is accumulated along each trajectory, the discount factor can be chosen as $\gamma=1$.
\end{remark}

The previous results establish that the original non-smooth reach-avoid probability can be represented exactly as an RL value function. 
To connect this RL formulation with the HJB framework, we next consider the smoothed reach-avoid problem introduced in \cref{sec:pde}, where the target set is replaced by $A_\epsilon$ and the auxiliary dynamics are defined accordingly:
\begin{align}
  v^{\mu}_{\mathrm{RL},\epsilon}(y)
  :=
  \mathbb{E}
  \left[
      \sum_{k=0}^{\infty}
      r_\epsilon(Y_k)
      \,\middle|\,
      Y_0=y,\mu
  \right],
  \label{eq:smoothed_rl_value_fun}
\end{align}
where the reward function is given by
\begin{align}
r_\epsilon(Y_k)
:=
\mI_{\mathcal T}(T_k)\,
l_\epsilon(\tilde X_k).
\label{eq:smoothed_reward}
\end{align}
The following theorem establishes that, in the limit as $\Delta t\to0$, the smoothed RL value function converges to the classical solution of the corresponding parabolic PDE.

\begin{theorem} \label{thm:pde_characterization}
Consider the augmented system~\eqref{eq:augmented_dynamics}
starting from an initial state
$y=[\tau,x^\top]^\top\in\mathcal{Y}$,
and the smoothed RL value function $v_{\mathrm{RL},\epsilon}^\mu$
defined in \eqref{eq:smoothed_rl_value_fun}.
Suppose that Assumptions~\ref{assumption:HJB}
and~\ref{assumption:smoothness} hold.
Then the following statements hold.
\begin{enumerate}
\item[(i)]
For each Markov policy $\mu$ satisfying the regularity conditions required for classical solvability,
the limit
\begin{align}
v_\epsilon^\mu(\tau,x)
:=
\lim_{\Delta t\rightarrow0}
v_{\mathrm{RL},\epsilon}^\mu(y)
\end{align}
exists uniformly on compact subsets of
$[0,T]\times\overline{O}_\epsilon$, where $O_\epsilon:=\mathbb{X}\setminus(A_\epsilon\cup B)$. 
Moreover,
$v_\epsilon^\mu$
is the unique classical solution of the linear parabolic PDE
\begin{align}
\mathcal{L}^\mu v_\epsilon^\mu(\tau,x)=0,
\quad
(\tau,x)\in(0,T]\times O_\epsilon,
\label{eq:linear_PDE}
\end{align}
with the initial and lateral boundary conditions
\begin{align}
v_\epsilon^\mu(\tau,x)
=
l_\epsilon(x),
\quad
(\tau,x)\in
\{0\}\times\overline{O}_\epsilon
\cup
(0,T]\times\partial O_\epsilon,
\label{eq:parabolic_boundary}
\end{align}
where $\mathcal{L}^\mu$ is induced by the Markov policy $\mu$ from the controlled differential operator $\mathcal{L}^u$ defined by
\begin{align}
\mathcal{L}^u\phi(\tau,x)
:=
&
-\partial_\tau\phi(\tau,x)
+
f(x,u)^\top\partial_x\phi(\tau,x)
\notag\\
&
+
\frac12
\operatorname{Tr}
\!\left(
\sigma(x,u)\sigma(x,u)^\top
\partial_x^2\phi(\tau,x)
\right).
\label{eq:parabolic_operator}
\end{align}

\item[(ii)]
Furthermore, the optimal value function satisfies
\begin{align}
v_\epsilon^\ast(\tau,x)
=
\lim_{\Delta t\rightarrow0}
\sup_{\mu\in\mathcal{U}_{\mathrm M}}
v_{\mathrm{RL},\epsilon}^\mu(y),
\label{eq:continuous_value_func}
\end{align}
where the limit is uniform on compact subsets of
$[0,T]\times\overline{O}_\epsilon$.
Moreover,
$v_\epsilon^\ast$
is the unique classical solution of the nonlinear parabolic PDE
\begin{align}
\sup_{u\in\mathbb{U}}
\mathcal{L}^u
v_\epsilon^\ast(\tau,x)
=
0,
\qquad
(\tau,x)\in(0,T]\times O_\epsilon,
\label{eq:value_PDE}
\end{align}
with the boundary condition~\eqref{eq:parabolic_boundary}.
\end{enumerate}
\end{theorem}

\begin{proof}
    See Appendix~\ref{proof:pde_characterization}.
\end{proof}

\begin{remark}
In Theorem\,\ref{thm:pde_characterization}, \eqref{eq:continuous_value_func} establishes that the optimal value function in the discrete-time RL setting  
converges to the solution of the continuous-time PDE~\eqref{eq:value_PDE} as $\Delta t \to 0$. 
The limiting value function $v^\ast_\epsilon(\tau,x)$ corresponds to the time-reversed form of the HJB solution  
$V_\epsilon(t,x)$ introduced in \secref{sec:pde}, via the relation
\begin{align}
v^\ast_\epsilon(\tau,x) = V_\epsilon(T - \tau, x).
\end{align}
As a result, the terminal condition at $t=T$ in \eqref{eq:HJB_BC} is transformed into an initial condition in \eqref{eq:parabolic_boundary}.  
\end{remark}

\section{Proposed Physics-informed RL Framework} \label{sec:PIRL_fixed}

Based on the RL formulation developed in the previous section, this section presents the proposed PIRL algorithm.
\Cref{sec:learing_objective} introduces the hybrid TD and PINN learning objective.
\Cref{sec:training_algorithm} presents the scheduled PIRL training algorithm.
Finally, Section~IV-C discusses the approximation guarantee of the learned value function.

\subsection{Hybrid TD-PINN Learning}
\label{sec:learing_objective}

Rather than solving the HJB equation using PINN-based policy iteration~\cite{Meng2024,Wang2024PIRL}, the proposed framework combines trajectory-based RL with PDE-constrained value-function approximation.
The key idea is to optimize the value function using both trajectory samples collected through interaction with the environment and collocation points sampled from the state space.
The former provides temporal-difference (TD) supervision not only for policy learning but also for guiding the optimization toward meaningful value-function approximations, whereas the latter enforces consistency with the governing PDE and its boundary conditions.
As demonstrated in \cref{sec:simulaion}, this trajectory-based supervision helps avoid undesirable failure modes that may arise when the PDE residual is minimized alone.

The proposed framework is applicable to a broad class of actor-critic algorithms.
The actor is responsible for policy improvement through interaction with the environment, whereas the critic approximates the reach-avoid value function.
In this work, we employ TD3 as a representative deterministic actor-critic algorithm, although the proposed learning framework is not restricted to this particular algorithm.
Let $Q_\theta(y,u)$ represent the critic network with parameters $\theta$, where the augmented state is given by $y=[\tau,x^\top]^\top$.
Following the TD3 framework, two critic networks $Q_{\theta_1}$ and $Q_{\theta_2}$ are maintained to mitigate overestimation bias. 
Each critic is trained using the temporal-difference loss: for $i = 1,\,2$,
\begin{align}
L_{\mathrm{TD},i}
=
\mathbb{E}{(y,u,r,y')}
\left[
\left(
Q_{\theta_i}(y,u)-Q_{\mathrm{target}}
\right)^2
\right],
\end{align}
and the target value is given by
\begin{align}
Q_{\mathrm{target}}
=
r
+(1-d)\gamma
\min_{i=1,2}
Q_{\theta_i^-}
\left(
y',
\pi_{\phi^-}(y')+\epsilon_\mathrm{n}
\right),
\end{align}
where $d$ is the terminal indicator, the networks $\pi_{\phi^-}$ and $Q_{\theta_i^-}$ denote the target actor and target critics, respectively, and $\epsilon_\mathrm{n}$ denotes the clipped noise added to the target action for target policy smoothing in TD3.

To enforce consistency with the governing PDE, collocation states $\{y^{(j)}\}_{j=1}^{N_{\mathrm{PDE}}}$ are sampled from a prescribed
distribution over the state domain. 
For each collocation state, the action used to evaluate the PDE residual is generated by perturbing the output of the target actor as
\begin{align}
u^{(j)}
=
\pi_{\phi^-}(y^{(j)})+\epsilon_{\mathrm{pde}}^{(j)},
\end{align}
where $\epsilon_{\mathrm{pde}}^{(j)}$ denotes clipped noise, analogous to target policy smoothing in TD3, introduced to stabilize the physics-informed critic update.
The corresponding PDE loss for each critic is defined as
\begin{align}
L_{\mathrm{PDE},i}
=
\frac{1}{N_{\mathrm{PDE}}}
\sum_{j=1}^{N_{\mathrm{PDE}}}
\left(
\mathcal{L}^{u^{(j)}}
Q_{\theta_i}(y^{(j)}, u^{(j)})
\right)^2,
\end{align}
where $\mathcal{L}^u$ stands for the differential operator defined in~\eqref{eq:parabolic_operator}.
In practice, the action components of the collocation points may be chosen from the current actor or from a neighborhood of the actor output, which improves the local consistency of the critic around the policy being learned.
In addition, the parabolic boundary condition of the smoothed
reach-avoid problem is enforced by introducing the boundary loss
\begin{align}
 L_{\mathrm{BC},i} 
=
\frac{1}{N_{\mathrm{BC}}}
\sum_{j=1}^{N_{\mathrm{BC}}}
\left(
 Q_{\theta_i}
\left(
y_{\mathrm{BC}}^{(j)},
\pi_{\phi^-}(y_{\mathrm{BC}}^{(j)})
\right)
-
l_\epsilon(x_{\mathrm{BC}}^{(j)})
\right)^2,
\end{align}
where
$y_{\mathrm{BC}}^{(j)}=(\tau_{\mathrm{BC}}^{(j)},x_{\mathrm{BC}}^{(j)})$
is sampled from the parabolic boundary
$\{0\}\times O_\epsilon\cup(0,T]\times\partial O_\epsilon$.

Consequently, the critic is trained by minimizing the weighted sum of these three objectives,
\begin{align}
L_{\mathrm{critic},i}
=
w_{\mathrm{TD}}
L_{\mathrm{TD},i}
+
w_{\mathrm{PDE}}
L_{\mathrm{PDE},i}
+
w_{\mathrm{BC}}
L_{\mathrm{BC},i}.
\label{eq:critic_loss}
\end{align}
The actor is updated according to the policy optimization step of the underlying RL algorithm, and the actor loss is given by
\begin{align}
L_{\mathrm{actor}}
=
-
\mathbb{E}_{y}
\left[
 Q_{\theta_1}
\left(
y,
\pi_\phi(y)
\right)
\right].
\label{eq:actor_loss}
\end{align}

\begin{algorithm}[t]
\caption{Scheduled Physics-Informed RL}
\label{alg:pirl}
\begin{algorithmic}[1]
\STATE \textbf{Input:} Initial actor $\pi_\phi$, critics $Q_{\theta_1},Q_{\theta_2}$, replay buffer $\mathcal{B}$, scheduling parameters $\{w_i^{\mathrm{init}},w_i^{\mathrm{final}},c,\alpha\}$, policy update frequency $d_\mathrm{up}$
\STATE Initialize target networks $\pi_{\phi^-}$, $Q_{\theta_1^-}$, $Q_{\theta_2^-}$
\FOR{$k=1,\ldots,K$}
    \STATE Receive trajectory samples from parallel rollout workers and store them in $\mathcal{B}$
    \STATE Sample a minibatch of transitions from $\mathcal{B}$
    \STATE Sample collocation points for the PDE residual and boundary losses
    \STATE Compute $L_{\mathrm{TD}}$, $L_{\mathrm{PDE}}$, and $L_{\mathrm{BC}}$
    \STATE Update critics by minimizing $L_{\mathrm{critic}}$ in~\eqref{eq:critic_loss}
    \IF{$k \equiv 0 \pmod {d_\mathrm{up}} $}
        \STATE Update actor using $L_{\mathrm{actor}}$ in~\eqref{eq:actor_loss}
        \STATE Update target networks
    \ENDIF
    \STATE Update $w_{\mathrm{TD}}(k)$, $w_{\mathrm{PDE}}(k)$, and $w_{\mathrm{BC}}(k)$
\ENDFOR
\end{algorithmic}
\end{algorithm}

\subsection{Scheduled Training Algorithm}
\label{sec:training_algorithm}

The hybrid objective in \eqref{eq:critic_loss} introduces a trade-off between trajectory-based RL and PDE-constrained value-function approximation.
At the early stage of training, the critic is typically inaccurate, making the trajectory-based TD supervision essential for learning an effective policy and obtaining a coarse approximation of the value function.
As training progresses, however, the PDE residual and boundary losses become increasingly important for improving the accuracy and physical consistency of the learned value function.
Motivated by this observation, we gradually shift the learning objective from RL-based supervision to PDE-based supervision during training.
Specifically, the loss weights are updated according to the sigmoid scheduling function
\begin{align}
\eta(k)
=
\frac{1}
{1+\exp\!\left[-\alpha(k-c)\right]},
\end{align}
where $k$ stands for the training iteration, $c$ is the center of the transition, and $\alpha$ determines its steepness.
The weight associated with each loss component is then given by
\begin{align}
w_i(k)
=
w_i^{\mathrm{init}}
+
\left(
w_i^{\mathrm{final}}
-
w_i^{\mathrm{init}}
\right)
\eta(k),
 \label{eq:wight_scheduling}
\end{align}
where $i\in\{\mathrm{TD},\mathrm{PDE},\mathrm{BC}\}$.
The parameters 
$w_i^{\mathrm{init}}$,
$w_i^{\mathrm{final}}$,
$c$, and
$\alpha$
determine the transition between trajectory-based and PDE-based supervision.
These parameters may be selected manually or determined automatically using an external hyperparameter optimization procedure.
In this work, the scheduling parameters are automatically adapted through an LLM-guided iterative search, as demonstrated in the drift example in \cref{sec:simulaion}.

The overall training procedure is summarized in Algorithm~\ref{alg:pirl}.
Trajectory samples are collected asynchronously by parallel rollout workers and accumulated in a replay buffer, while the learner updates the actor and critic networks independently.
At each learning iteration, a minibatch of transitions is sampled from the replay buffer to evaluate the TD loss, and collocation points are independently sampled from the state-time domain to evaluate the PDE residual and boundary losses.
The critic is updated by minimizing the hybrid objective in~\eqref{eq:critic_loss}.
Following the underlying actor-critic algorithm, the actor and target networks are updated only once every $d_\mathrm{up}$ critic updates.
Finally, the loss weights are updated according to the scheduling rule described above, and the resulting weights are used in the next learning iteration.

\begin{remark}
Condition \eqref{eq:value_PDE} corresponds to the Bellman optimality equation, where purely sample-based evaluations are known to induce overestimation bias under function approximation. While our framework mitigates this bias using TD3 and PDE residual regularization, other existing techniques can be modularly integrated during the initial phase. For instance, one could smoothly interpolate between the expectation and max operators~\cite{omura2025gradual}, building on the expectile parameter of Implicit Q-Learning~\cite{kostrikov2022iql}, or via soft non-expansion operators such as Mellowmax~\cite{asadi2017alternative}. 
\end{remark}

\subsection{Approximation Guarantee}
\label{sec:guarantee}

A key advantage of incorporating physics-informed constraints into the critic network is that the approximation accuracy can be related to the governing PDE.
Recall from Theorem~\ref{thm:pde_characterization} that the
optimal smoothed reach-avoid value function $v_\epsilon^*$ is
the unique classical solution of the PDE~\eqref{eq:value_PDE} with the boundary condition
\eqref{eq:parabolic_boundary}.
For notational simplicity, define the parabolic domain and its
boundary by
\begin{align}
\mathcal D_\epsilon
&:=
(0,T)\times O_\epsilon,
\\
\partial_p\mathcal D_\epsilon
&:=
\{0\}\times O_\epsilon
\cup
(0,T]\times\partial O_\epsilon .
\end{align}
Let $v_\theta(\tau,x)$ be the approximated value function represented by the critic and actor networks after training.
Applying the PINN error analysis of \cite{Wang2026} to the policy-evaluation PDE associated with the learned policy yields the following approximation-error bound.

\begin{corollary} \label{thm:error_bound}
Suppose $v_\theta
\in
C^{1,2}(\mathcal D_\epsilon)
\cap C^0(\overline{\mathcal D}_\epsilon)$, and let
$Z_{\rm int}$ and $Z_{\rm bd}$ be uniformly distributed over
$\mathcal D_\epsilon$ and $\partial_p\mathcal D_\epsilon$,
respectively.
Assume that
\begin{align}
\mathbb E
\left[
\left|
v_\theta(Z_{\rm bd})-l_\epsilon(x_{\rm bd})
\right|
\right]
&\le
\delta_{\rm BC},
\\
\mathbb E
\left[
\left|
\mathcal L^{\mu_\phi}v_\theta(Z_{\rm int})
\right|
\right]
&\le
\delta_{\rm PDE},
\end{align}
where $Z_{\rm bd}=(\tau_{\rm bd},x_{\rm bd})$.
Assume further that
$v_\theta$ and
$\mathcal L^{\mu_\phi}v_\theta$
are Lipschitz continuous on
$\overline{\mathcal D}_\epsilon$ and
$\mathcal D_\epsilon$, respectively.
Then there exists a constant $C>0$, depending only on the
geometry of the domain and the regularity of the parabolic
operator, such that
\begin{align}
v_\theta(\tau,x)
\le
v_\epsilon^*(\tau,x)
+
\widetilde\delta_{\rm BC}
+
C\widetilde\delta_{\rm PDE},
\quad
(\tau,x)\in\mathcal D_\epsilon,
\label{eq:error_upper_bound}
\end{align}
where $\widetilde\delta_{\rm BC}$ and
$\widetilde\delta_{\rm PDE}$ are constants obtained from
$\delta_{\rm BC}$ and $\delta_{\rm PDE}$ through
$L^1$-to-$L^\infty$ estimate for Lipschitz functions under
uniform sampling.
Furthermore, if the learned policy is optimal, i.e.,
$\mu_\phi=\mu^*$, then
\begin{align}
\left|
v_\theta(\tau,x)
-
v_\epsilon^*(\tau,x)
\right|
\le
\widetilde\delta_{\rm BC}
+
C\widetilde\delta_{\rm PDE},
\quad
(\tau,x)\in\mathcal D_\epsilon .
\label{eq:two_side_bound}
\end{align}
\end{corollary}

\begin{proof}
    See Appendix~\ref{proof:error_bound}.
\end{proof}

\begin{remark}
The inequality \eqref{eq:error_upper_bound} bounds the learned value relative to the optimal reach-avoid value $v_\epsilon^*$.
A complementary two-sided guarantee can be obtained for the value of the learned policy itself.
Let $v_\epsilon^{\mu_\phi}$ denote the reach-avoid value function under $\mu_\phi$.
By applying the PINN error estimate in~\cite{Wang2026} to the linear policy-evaluation PDE associated with $\mu_\phi$, the assumptions in Corollary~\ref{thm:error_bound} imply
\begin{align}
\left|v_\theta(\tau,x)-v_\epsilon^{\mu_\phi}(\tau,x)\right|
\leq
\widetilde{\delta}_{\mathrm{BC}}
+C\widetilde{\delta}_{\mathrm{PDE}},
\quad (\tau,x)\in D_\epsilon.
\end{align}
Thus, even when the learned policy is suboptimal, sufficiently small
PDE residual and boundary mismatch guarantee that the critic
accurately approximates the reach-avoid probability achieved by
that policy.
\end{remark}

The above corollary provides a theoretical justification for incorporating the PDE residual into critic training.
This approximation guarantee is unavailable in standard RL, where the critic is supervised only through temporal-difference targets.
Note that the TD loss introduced in the proposed PIRL framework is not required for the above error bound.
Instead, its role is to improve the optimization process by providing trajectory-based supervision, thereby mitigating poor local minima that may satisfy the PDE residual approximately while failing to recover the correct value function.

\section{Numerical experiments}
\label{sec:simulaion}

This section evaluates the proposed PIRL framework through
numerical experiments. 
We first consider a one-dimensional stochastic reachability
problem in \cref{sec:1d}, for which an analytical solution is
available. This example enables a direct comparison of the
learned value functions and is used to investigate the learning
behavior of TD3, PINN-based policy iteration, and the proposed
PIRL. 
We then consider a reach-avoid control problem for vehicle
drifting in \cref{sec:drift}, which demonstrates the applicability
of the proposed method to a higher-dimensional and practically
challenging control problem.
The source code and numerical results are publicly available at
\url{https://github.com/hoshino06/ProbReachPIRL}.


\subsection{1D Reachability Problem}
\label{sec:1d}

We first consider a one-dimensional stochastic reachability problem to illustrate the learning behavior of the proposed method in a setting where the ground-truth solution is available.
The system dynamics are given by
\begin{align}
\mathrm{d}X_t = U_t \mathrm{d}t + \mathrm{d}W_t,
\end{align}
where $U_t \in [-1,1]$ is the control input and $W_t$ is a standard Brownian motion. 
We consider a finite-horizon reachability problem with a time horizon of $T=2$, in which the target set is defined as
\begin{align}
A = \{x \in \mathbb{R} \mid x \ge 2\}.
\end{align}
No avoid set is introduced in this example, and therefore the problem reduces to a pure reachability problem. Since moving toward larger values of $x$ always increases the probability of reaching the target, the optimal policy is simply $u^\ast(x)=1$ for all states. 
The corresponding reachability probability can be obtained analytically \cite{chern21}: 
\begin{align}
v(\tau, x)
&=\mathbb{P}(\tau_A\le \tau\mid X_0=x),   
\end{align}
where $\tau_A=\inf\{t\ge0\mid X_t\ge2\}$ is the first hitting time of the target set.

The accuracy of the proposed PIRL framework was evaluated by comparing the learned value function with the analytical solution.
The actor had two hidden layers with 32 neurons each, while the critics had three hidden layers with 32 neurons each.
The networks were trained using the Adam optimizer with a learning rate of $10^{-4}$.
For computing the PDE residual and boundary losses, 1000 interior collocation points and 100 points on each boundary were sampled at every update. 
The loss weights were scheduled according to \cref{eq:wight_scheduling}, gradually transitioning from $(w_{\mathrm{TD}},w_{\mathrm{HJB}},w_{\mathrm{BC}})=(1,0,0)$ to $(0,1,1)$ with $c=5\times10^4$ and $\alpha=10^{-4}$.
\Cref{fig:1D_learned_surface} compares the value function obtained after $2\times10^5$ update steps with the analytical solution.
The learned value function closely matches the ground-truth solution over the entire state-time domain except in the vicinity of the singularity of the ground truth at $(t,x)=(0,2)$. 
The heatmap of the absolute error further confirms that the approximation error remains small throughout the domain away from this singular region. These results verify that the proposed framework successfully learns both the reachability probability and the associated optimal policy.

\begin{figure}
    \centering
    \begin{minipage}[t]{0.45\linewidth}
        \centering
        \includegraphics[width=\linewidth]{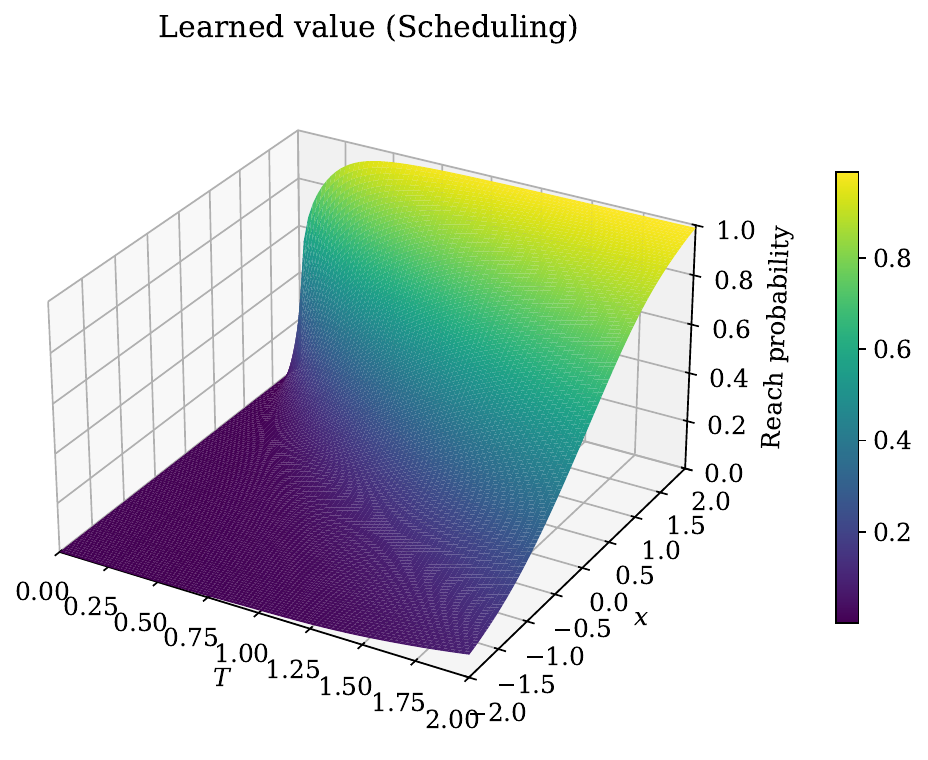}
        \subcaption{Learned value function}\label{fig:dc_coupled_colocated}
    \end{minipage}\hfill
    \begin{minipage}[t]{0.45\linewidth}
        \centering
        \includegraphics[width=\linewidth]{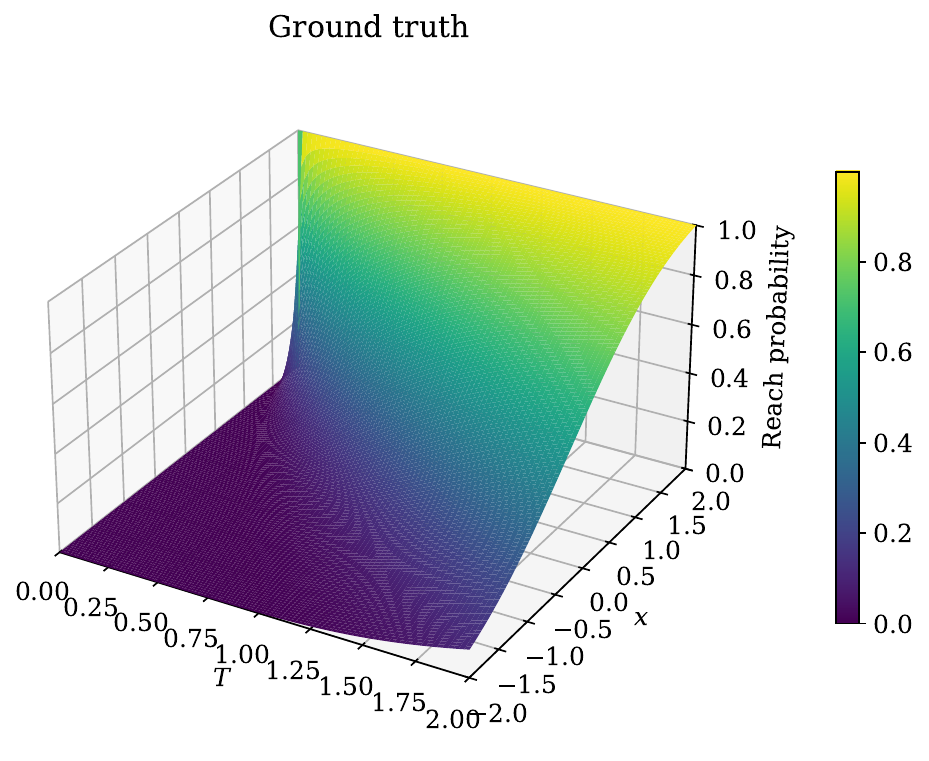}
        \subcaption{Ground truth}
        \label{fig:surface_ground_truth}
    \end{minipage}\\[2mm]
    \begin{minipage}[t]{1.0\linewidth}
    \centering
    \includegraphics[width=\linewidth]{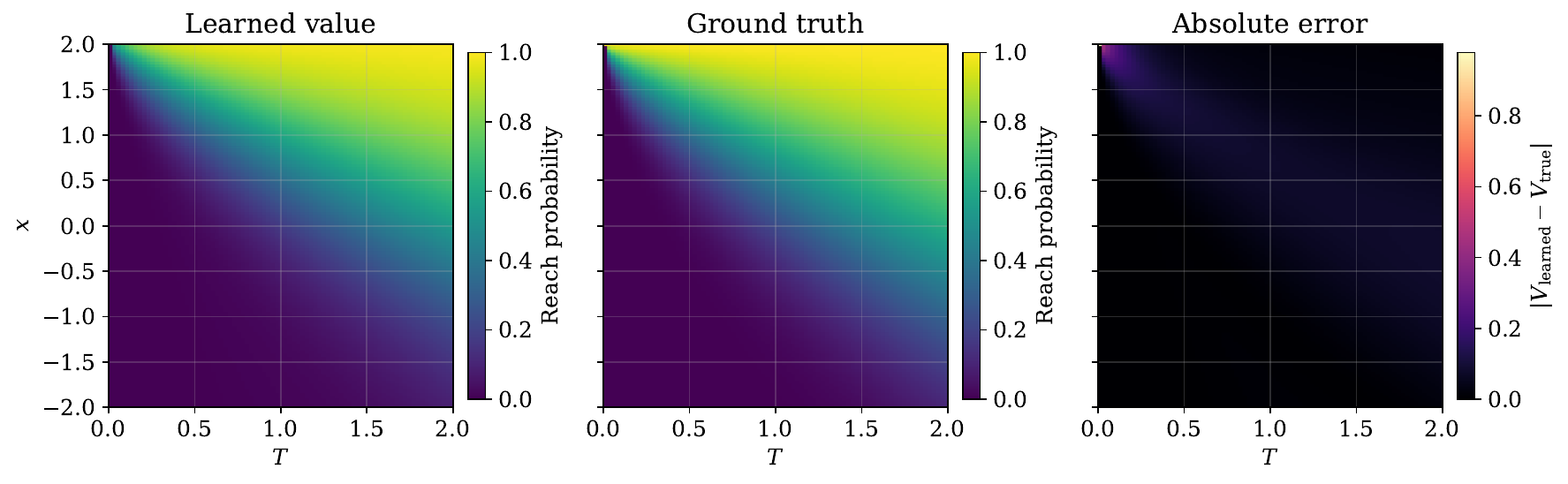}
    \subcaption{Heatmaps of learned value, ground truth, and absolute error}
    \label{fig:surface_ground_truth}
    \end{minipage}
    \caption{Learned value function for 1D reachability problem}
    \label{fig:1D_learned_surface}
\end{figure}

To further investigate the learning characteristics of TD3, PINN-based policy iteration (PINN-PI)~\cite{Meng2024,Wang2024PIRL}, and the proposed PIRL, they were evaluated over ten random seeds. 
\Cref{fig:1D_performance} compares the control performance and value-function accuracy obtained by the three approaches.
As shown in \cref{fig:1D_reward}, TD3 successfully learns the task for all ten runs.
By contrast, PINN-PI exhibited a pronounced failure mode: four runs achieved rewards comparable to those of TD3, whereas the remaining six converged to solutions with substantially lower rewards.
The proposed PIRL achieved near-optimal rewards in all ten runs, indicating that the inclusion of trajectory-based TD learning prevents the collapse observed in PINN-PI.
The mean-squared errors with respect to the analytical value function are shown in \cref{fig:1D_mse}.
Although TD3 reliably learned the optimal policy, its value-function error was larger than those of the successful PINN-PI and PIRL runs.
The proposed PIRL achieved an error comparable to that of successful PINN-PI while retaining the reliable policy learning observed in TD3.
This relationship is summarized more directly in \cref{fig:1D_reward_mse}, where each small marker represents an individual run and each large marker represents the corresponding group average.
The proposed PIRL consistently occupies the lower-right region, demonstrating that it combines the policy-learning robustness of TD3 with the value-function accuracy of successful PINN-PI training.

\begin{figure}[t]
    \centering
    \begin{minipage}[t]{0.49\linewidth}
        \centering
        \includegraphics[width=0.8\linewidth]{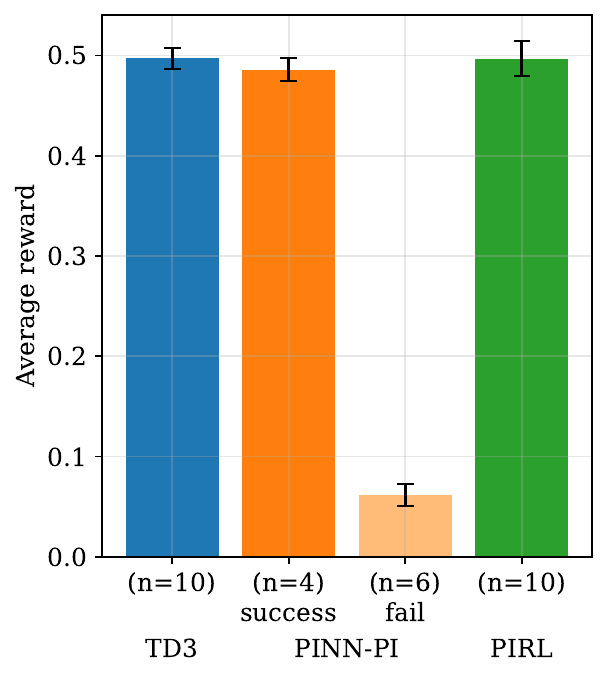}
        \subcaption{Average rewards}
        \label{fig:1D_reward}
    \end{minipage}
    \hfill
    \begin{minipage}[t]{0.49\linewidth}
        \centering
        \includegraphics[width=0.8\linewidth]{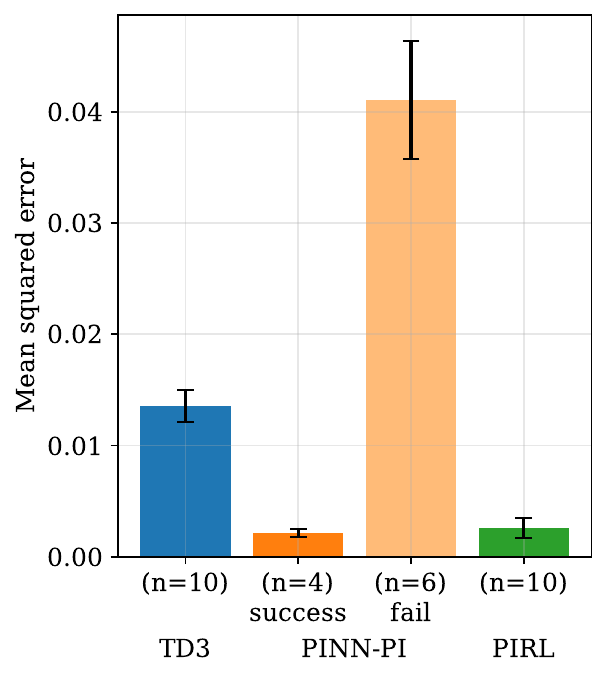}
        \subcaption{Mean-squared errors}
        \label{fig:1D_mse}
    \end{minipage}\\[2mm]
    \begin{minipage}[t]{\linewidth}
        \centering
        \includegraphics[width=0.6\linewidth]{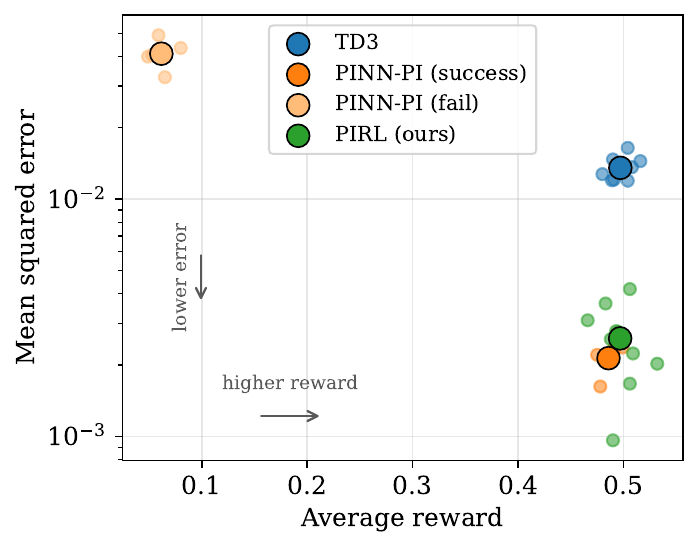}
        \subcaption{Average reward vs. mean-squared error}
        \label{fig:1D_reward_mse}
    \end{minipage}
    \caption{Comparison of performance over ten random seeds}
    \label{fig:1D_performance}
\end{figure}

\begin{figure}[t]
    \centering
    \includegraphics[width=0.9\linewidth]{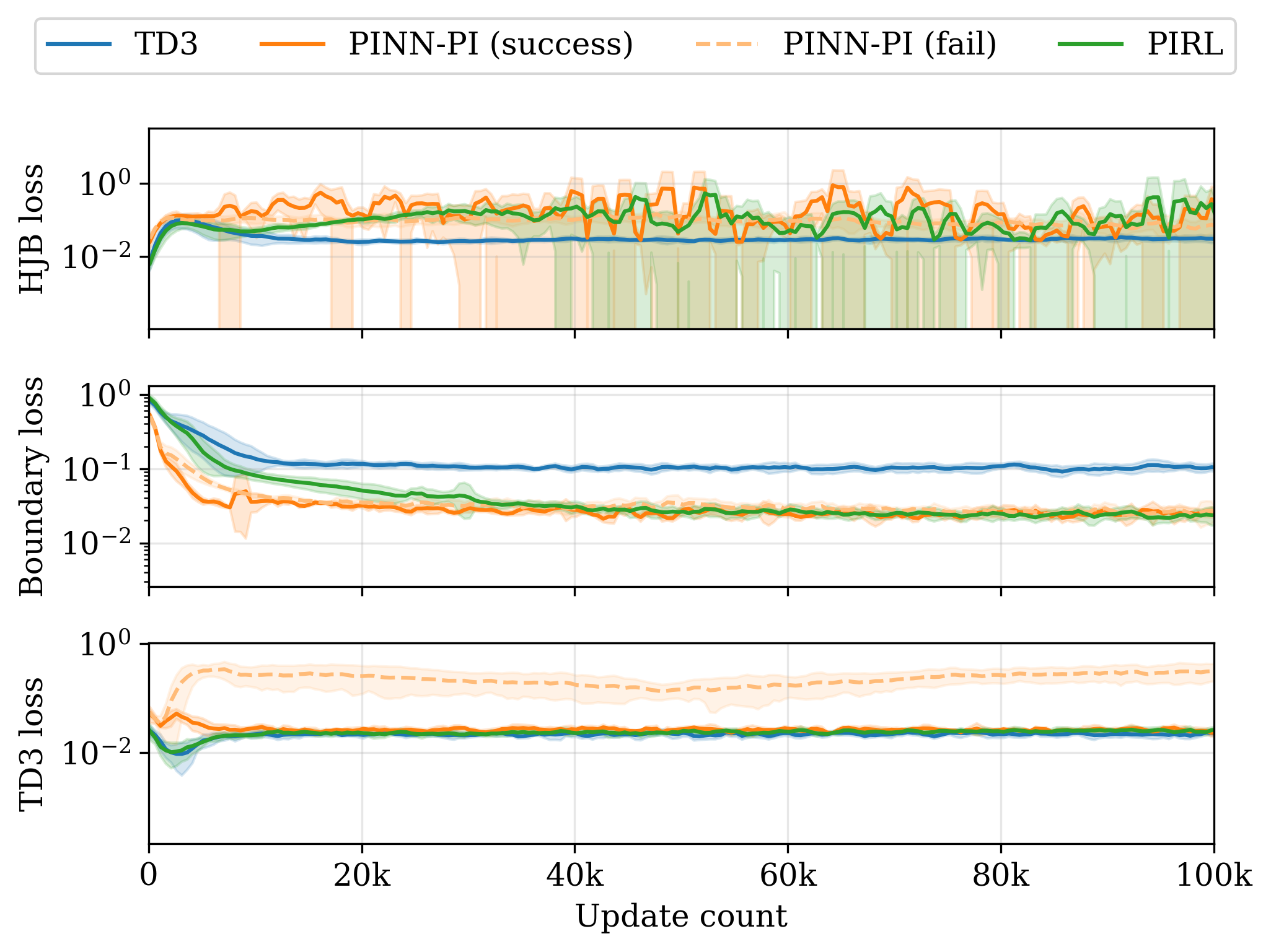}
    \caption{Evolution of losses during training}
    \label{fig:1D_loss_terms}
\end{figure}

To examine the mechanism underlying these differences, \cref{fig:1D_loss_terms} compares the evolution of the HJB, boundary, and TD losses.
The HJB losses of the failed PINN-PI runs are not larger than those of the successful runs.
Similarly, their boundary losses decrease to levels comparable to those obtained in successful PINN-PI training.
Thus, small physics-informed losses do not necessarily imply that the learned value function is meaningful.
PDE residual-based optimization can instead converge to a nearly trivial value function that approximately satisfies the sampled HJB residual and boundary conditions while failing to recover the correct reachability structure.
The distinction becomes clearer in the TD loss.
When all trained networks are evaluated using trajectory-based TD targets, the failed PINN-PI runs exhibit a substantially larger TD loss than the other groups.
Their critics are therefore inconsistent with the returns observed along sampled trajectories, even though their HJB and boundary losses are not large.
The TD objective provides complementary information that distinguishes these inaccurate solutions and guides learning toward a value function associated with an effective policy.
By first using TD learning to establish the coarse reachability structure and then introducing the physics-informed losses, the proposed scheduling strategy avoids the failure mode of PINN-PI while reducing the value-function error of pure TD3.

\subsection{Reach-Avoid Control of Vehicle Drifting}
\label{sec:drift}

\subsubsection{Problem setup}

We next consider a stochastic reach-avoid control problem for autonomous vehicle drifting.
The vehicle is modeled by a nonlinear path-relative dynamic bicycle model adapted from~\cite{Hindiyeh2014}, which describes both the vehicle motion and its evolution relative to a reference path. 
The lateral tire forces are represented by a nonlinear Fiala tire model, including the reduction in the rear lateral-force capacity caused by the longitudinal driving force.
The state is given by
\begin{equation}
x =
\begin{bmatrix}
e_y & e_\psi & v_x & v_y & r & \delta & \mu
\end{bmatrix}^{\top},
\end{equation}
where $e_y$ and $e_\psi$ stand for the lateral position and heading errors with respect to the reference path, $v_x$ and $v_y$ are the longitudinal and lateral velocities, $r$ is the yaw rate, $\delta$ is the steering angle, and $\mu$ is the road friction coefficient.
The friction coefficient $\mu$ is treated as an episode-wise constant parameter and included in the state, allowing the policy to account for variations in road friction within a single policy.
In this study, $\mu$ was varied over the range $[0.5,0.6]$ during training.
The remaining horizon $\tau$ is augmented to the state as described in \cref{sec:model_free_RL}.
The control input is
\begin{equation}
u =
\begin{bmatrix}
\dot{\delta} & F_x
\end{bmatrix}^{\top},
\end{equation}
where $\dot{\delta}$ is the steering-rate command and $F_x$ is the longitudinal force. 

To formulate a stochastic reach-avoid problem, additive Brownian disturbances are introduced into the vehicle dynamics.
The disturbances represent modeling errors and external uncertainties affecting the vehicle motion during drifting.
The target set is a neighborhood of a drift equilibrium. Following \cite{Hindiyeh2014}, the equilibrium is specified by a desired longitudinal speed, sideslip angle, and friction coefficient, and the remaining equilibrium quantities are computed from the steady-state model equations. 
The avoid set represents loss-of-control conditions, including road departure, excessive heading error, excessive sideslip, excessive yaw rate, loss of forward speed, steering saturation, and friction values outside the admissible range. 
This construction yields a physically meaningful reach-avoid problem: the controller must drive the vehicle toward
an unstable drift operating point while avoiding states that
correspond to spinout, departure, or loss of control.

\subsubsection{Training process}

We first trained a TD3 agent with an actor consisting of two hidden layers of 32 neurons each and critics consisting of three hidden layers of 64 neurons each.
In preliminary experiments, however, starting TD3 training directly from a broad initial-state distribution in the entire eight-dimensional hypercube was not effective, because most rollouts terminated without reaching the drift target and therefore provided little information about recoverable trajectories.
We therefore employed a curriculum over the initial-state distribution.
The spread of the sampled initial states around the drift equilibrium was controlled by a scalar sampling scale.
TD3 was first trained with sampling scale $0.4$ for $10^6$ learner updates so that the policy learned to recover from local perturbations around the drift equilibrium.
The sampling scale was then increased to $1.0$, and the initial states were sampled from a mixture of the $\beta$--$r$ plane, the $e_y$--$e_\psi$ plane, and the full state space, with mixture probabilities $0.3$, $0.3$, and $0.4$, respectively.
This second TD3 stage was continued until $5\times10^6$ updates, allowing the policy to gradually extend its recovery behavior toward the reach-avoid boundary.
\begin{figure}
    \centering
    \includegraphics[width=0.8\linewidth]{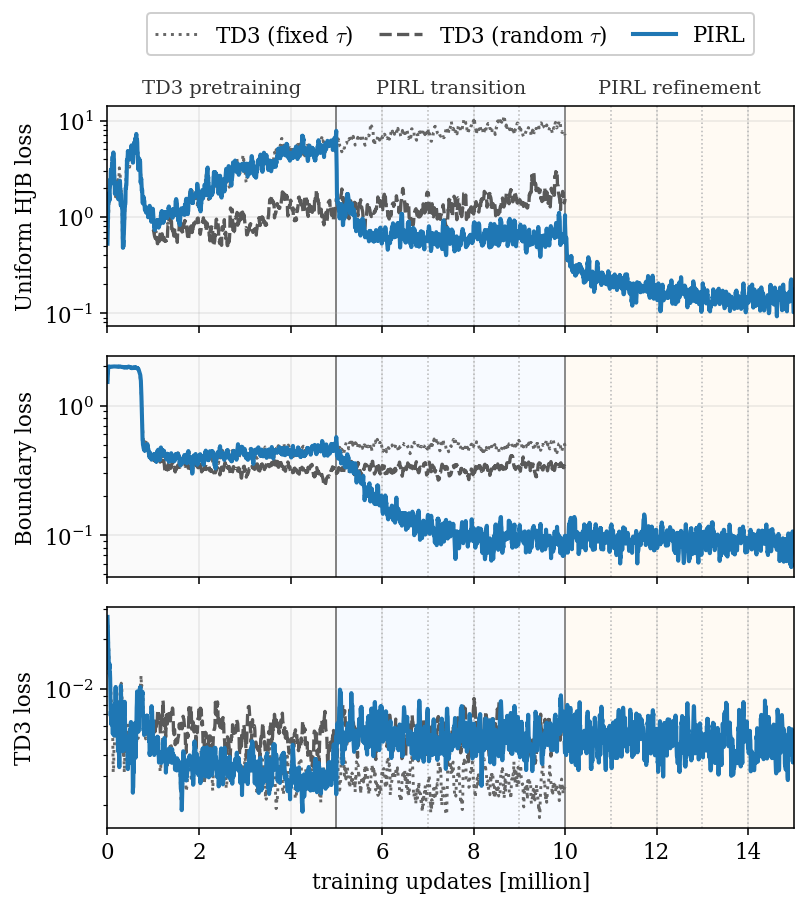}
    \caption{Evolution of losses in drifting example}
    \label{fig:drift_training_loss}
    \vspace{-3mm}
\end{figure}
From a resulting TD3 checkpoint, we continued training with the PIRL objective.
The continuation consisted of a PIRL transition stage from $5\times10^6$ to $10^7$ updates followed by a refinement stage, as illustrated in \cref{fig:drift_training_loss}.
The continuation was controlled by an outer-loop optimization based on an LLM-based coding agent (OpenAI codex was used in our experiment). 
At each round, a Python orchestration script summarized the completed training runs, constructed a scheduling prompt, invoked the coding agent, parsed the resulting scheduling decision, and launched several candidate PIRL continuations.
Each candidate was trained for $10^6$ updates before the next scheduling decision.
The scheduling prompt combined a fixed template with an interactively updated context block providing phase-specific guidance, together with a summary of the accumulated experimental results.
The prompt template used for the experiment is provided in Appendix~\ref{app:codex_prompt}.

The LLM coding agent was guided by several qualitative observations obtained from preliminary manual trials and preceding outer-loop experiments.
Throughout the PIRL continuation, the TD loss was kept active with $w_{\mathrm{TD}}=1$, while $w_{\mathrm{HJB}}$ and $w_{\mathrm{BC}}$ were introduced gradually.
Increasing the physics-informed weights too rapidly was found to degrade the policy and, in some cases, drive the critic toward a nearly trivial value function, as in the failure mode observed in \cref{sec:1d}.
The scheduling candidates were therefore designed to increase the physics-informed weights conservatively while preserving the average reward and the Monte Carlo estimate of the reach-avoid probability achieved by the TD3 checkpoint.
Preliminary trials also showed that the initialization of the remaining horizon $\tau$ in RL rollouts was important. 
We evaluated TD3 checkpoints trained with either a fixed initial horizon $\tau=T$ or horizons randomized over $[0,T]$. 
The fixed-horizon TD3 achieved better control performance and was therefore used to initialize PIRL as shown in \cref{fig:drift_training_loss}. 
However, continuing PIRL with the same fixed-horizon rollout  made the transition less stable. 
In this case, the TD samples were concentrated along trajectories initialized at $\tau=T$, whereas the HJB residual was evaluated over the broader time-augmented domain. 
This mismatch between the trajectory-based and physics-informed supervision became pronounced as the HJB weight increased and  led to critic collapse and degraded policy performance. 
We therefore randomized the initial horizon over $[0,T]$ during PIRL continuation, thereby broadening the TD-sample distribution and improving its overlap with the domain on which the HJB residual was enforced.

The distinction between the transition and refinement stages in \cref{fig:drift_training_loss} primarily reflects how the HJB collocation distribution was expanded during the outer-loop search.
In the transition stage, the HJB collocation points were sampled mainly from the replay-buffer distribution, because the policy and critic inherited from TD3 were reliable primarily around the regions visited by the closed-loop trajectories.
Applying the HJB loss uniformly over the full state domain at this stage introduced strong supervision in regions not yet supported by trajectory data and often destabilized the continuation.
Replay-based collocation instead imposed local PDE consistency around the reach-avoid structure already established by TD learning, allowing the physics-informed weights to be increased without abruptly altering the learned policy-value relationship.
Once the average reward and Monte Carlo reach-avoid probability had stabilized, the outer-loop search entered the refinement stage, in which the collocation distribution was gradually expanded away from the replay-buffer samples and uniformly sampled points were increasingly incorporated.
This expansion was intended to improve PDE consistency in regions less frequently visited by the current policy while retaining the control performance obtained during the transition.
As shown in \cref{fig:drift_training_loss}, the HJB loss decreases substantially during the refinement stage, whereas the TD loss remains at approximately the same level and the boundary loss remains small, indicating that the critic's PDE and boundary consistency was improved without substantially disrupting the trajectory-based value structure.

\subsubsection{Learning results}

\begin{figure}[t]
  \centering
  \begin{subfigure}[t]{0.47\linewidth}
    \centering
    \includegraphics[width=\linewidth]{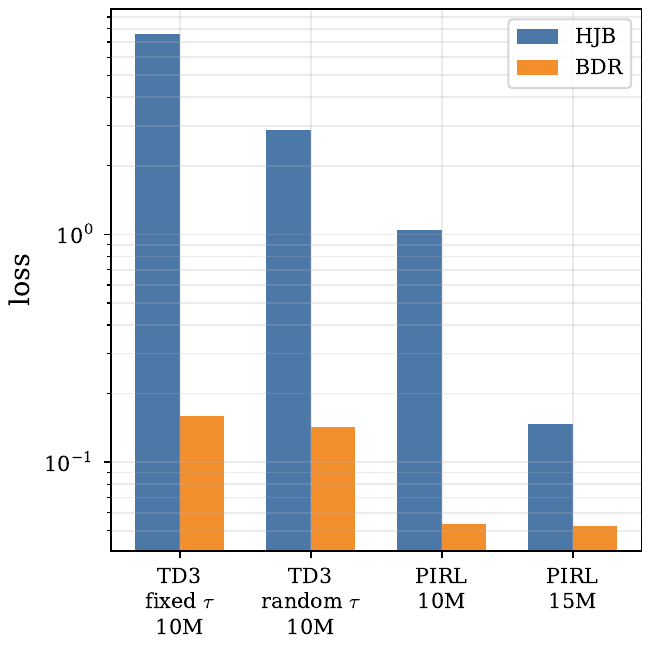}
    \caption{HJB and boundary loss}
    \label{fig:learning_results:a}
  \end{subfigure}\hfill
  \begin{subfigure}[t]{0.52\linewidth}
    \centering
    \includegraphics[width=\linewidth]{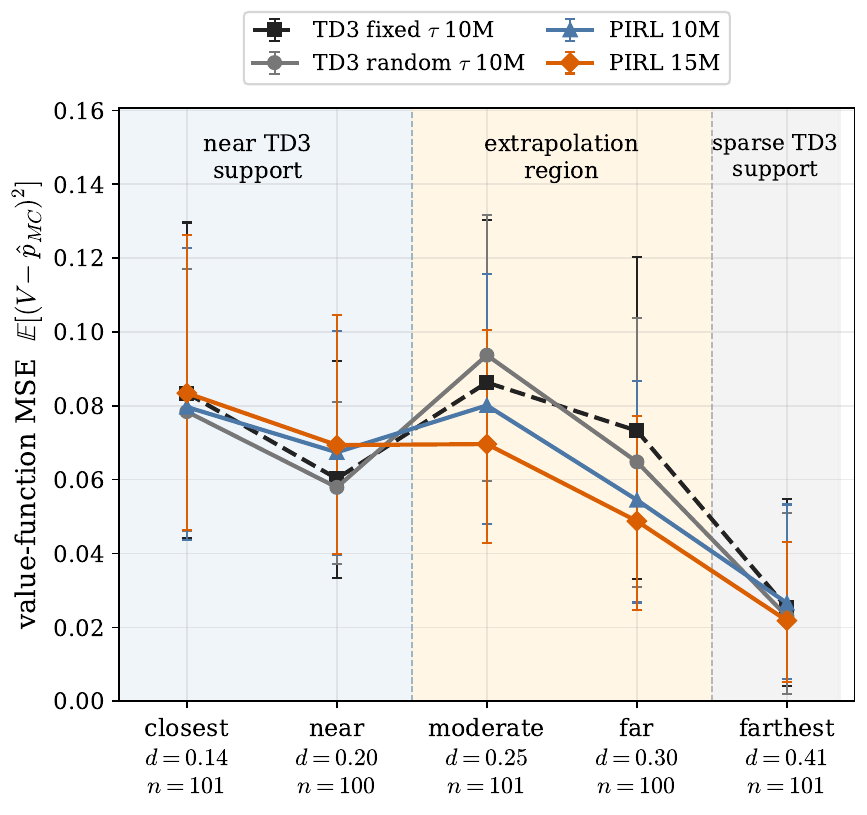}
    \caption{Value-function accuracy}
    \label{fig:learning_results:b}
  \end{subfigure}
  \caption{PDE consistency and value-function accuracy}
  \label{fig:drift-error-validation}
\end{figure}

\Cref{fig:drift-error-validation} summarizes the validation results for the learned value functions.
\Cref{fig:learning_results:a} compares the validation HJB and boundary-condition losses of the fixed- and random-horizon TD3 baselines with those of PIRL at the end of the transition stage ($10$M updates)
and after refinement ($15$M updates).
PIRL achieves lower HJB and boundary-condition losses than both TD3
baselines, and the refinement stage further improves HJB consistency.
\Cref{fig:learning_results:b} evaluates the accuracy of the learned value by comparing it with the empirical reach-avoid probability estimated from Monte Carlo closed-loop simulations under each trained policy.
The evaluation states are grouped according to their distance from the states visited  by the TD3 policy.
The closest and near groups correspond to regions well covered by TD3 trajectories, whereas the moderate and far groups require increasing extrapolation beyond the TD3 visitation distribution.
Near the TD3 trajectories, the TD3 baselines achieve lower value-function errors, whereas the advantage of PIRL becomes apparent as the evaluation states move away from the TD3 visitation distribution. 
In particular, PIRL at $15$M updates achieves the lowest MSE in the moderate
and far groups.
When averaged over all evaluation states, the MSE values are $0.0657$,
$0.0636$, $0.0617$, and $0.0586$ for fixed-horizon TD3, random-horizon TD3, PIRL at $10$M, and PIRL at $15$M, respectively.
The errors decrease in the farthest group for all four models, because the reach-avoid probability is close to zero for many states in this region and is therefore relatively easy to approximate. 
Overall, PIRL improves value-function accuracy, particularly outside the
region well supported by TD3 trajectories, suggesting improved
generalization beyond the TD3 visitation distribution.

Finally, to provide a physical interpretation of the learned value function, \cref{fig:drift_value_contours} shows the reachability value learned by PIRL at $15$M updates on two representative two-dimensional slices of the state space.
The arrows represent the projected closed-loop vector field.
Such phase-portrait visualizations are widely used in vehicle dynamics because they provide an intuitive picture of the direction of motion,
equilibrium structure, and stable operating regions in low-dimensional
projections~\cite{Bobier2019}.
In both the $\beta$--$r$ and $e_y$--$e_\psi$ planes, the learned value function forms a high-value region around the drift equilibrium and extends along the direction suggested by the projected closed-loop vector field.
However, the reachability value cannot be fully interpreted from a two-dimensional projected vector field, because reachability also depends on the coupled evolution of the remaining state variables and the avoid-set constraints.
Consequently, even when the projected arrows point toward a high-value region on a given slice, the learned value does not necessarily become large.
\Cref{fig:drift_results} complements this interpretation by showing stochastic closed-loop rollouts from three representative initial conditions in the $\beta$--$r$ plane.
The resulting trajectories are projected onto the $\beta$--$r$ and $e_y$--$e_\psi$ planes and overlaid on the learned value contours.
The comparison between the two projections reveals recovery and failure mechanisms that cannot be inferred from a single projected phase portrait.
For example, the orange initial state appears recoverable in the $\beta$--$r$ plane because the projected trajectories move toward the high-value region. 
In the $e_y$--$e_\psi$ plane, however, the same trajectories leave the path-relative region from which the drift target can be reached.
This explains why the learned value assigned to this initial condition does not become large, even though the projected flow in the $\beta$--$r$ plane points toward the high-value region.
These results show that the learned critic appropriately captures the reach-avoid structure of the closed-loop phase space, including aspects that cannot be inferred from individual low-dimensional projections in classical vehicle dynamics analysis.



\begin{figure}[t]
    \centering
    \begin{subfigure}[t]{0.49\linewidth}
        \centering
        \includegraphics[width=\linewidth]{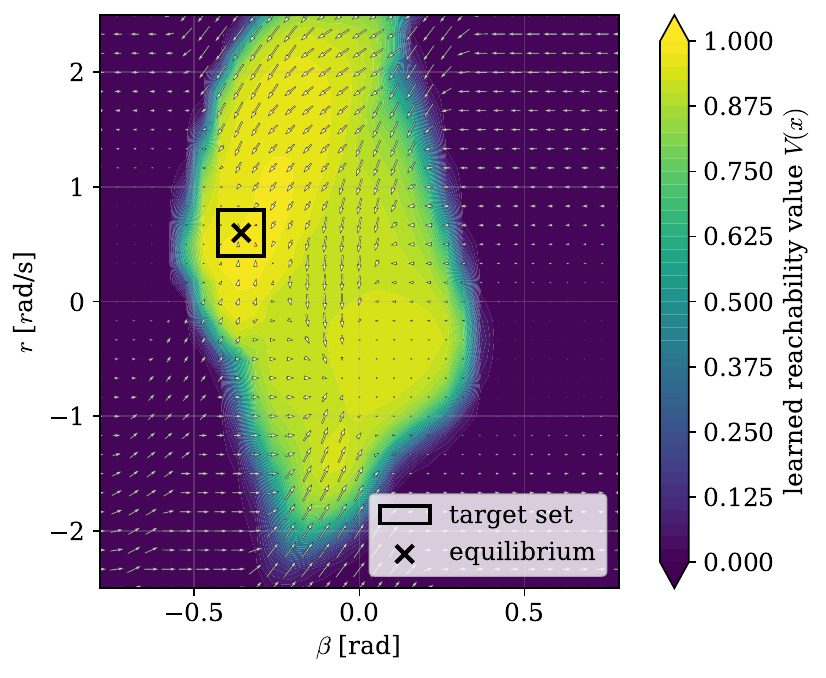}
        \caption{$\beta$--$r$ plane}
        \label{fig:drift_value_contour_beta_r}
    \end{subfigure}
    \hfill
    \begin{subfigure}[t]{0.49\linewidth}
        \centering
        \includegraphics[width=\linewidth]{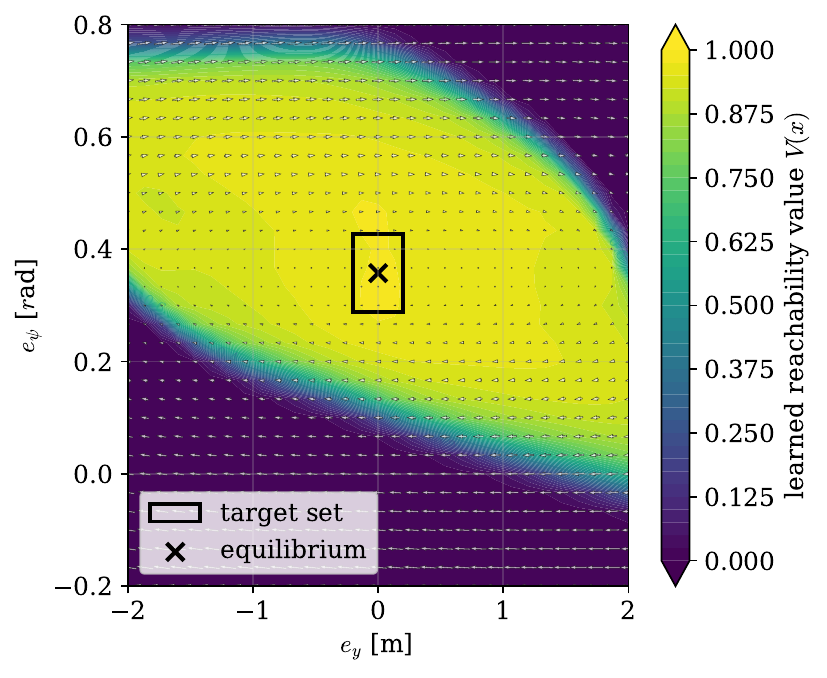}
        \caption{$e_y$--$e_\psi$ plane}
        \label{fig:drift_value_contour_ey_epsi}
    \end{subfigure}
    \caption{Learned value on two-dimensional slices with $\mu=0.55$}
    \label{fig:drift_value_contours}
\end{figure}

\begin{figure}[t]
    \centering
    \begin{subfigure}[t]{0.46\linewidth}
        \centering
        \includegraphics[width=\linewidth]{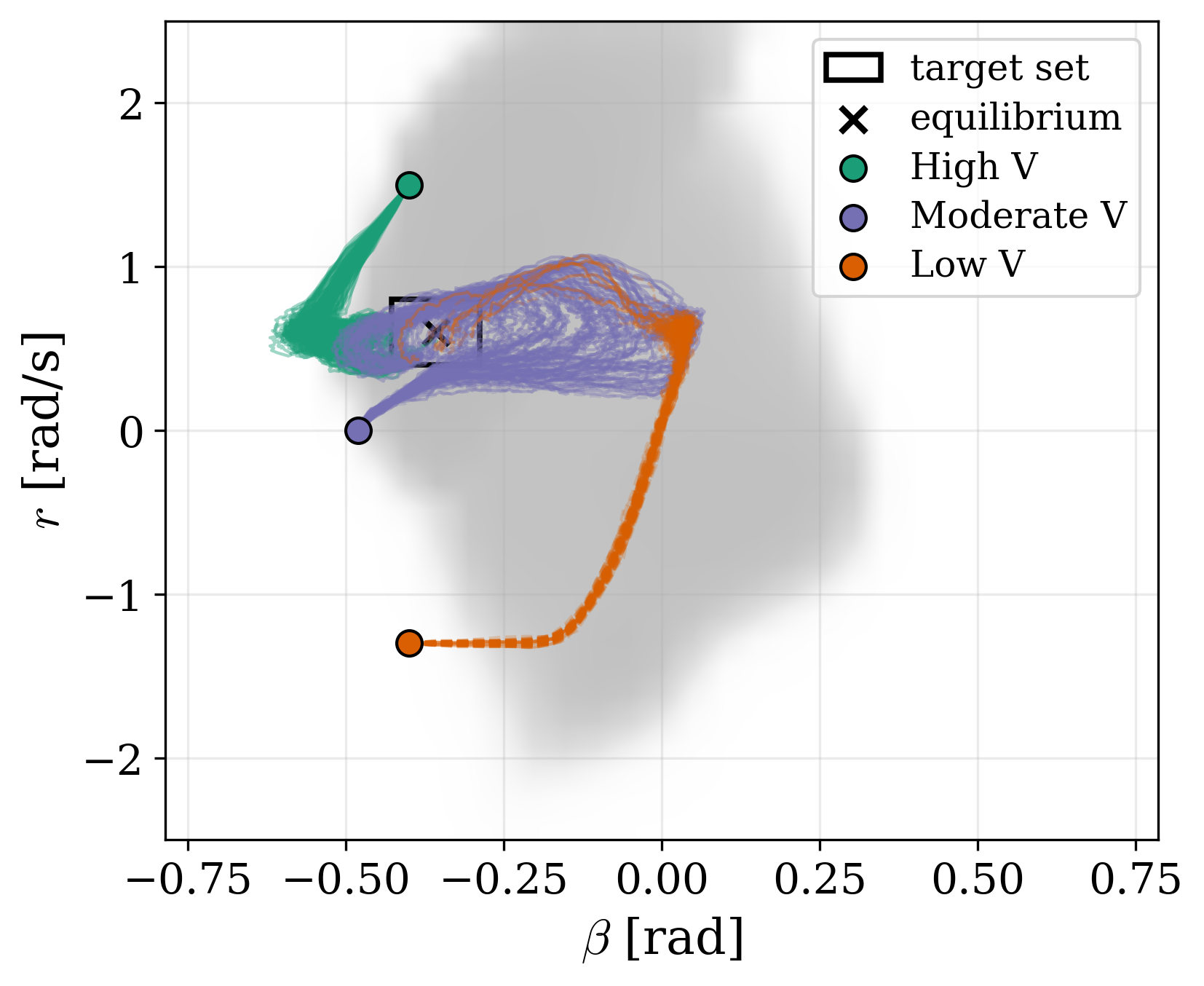}
        \caption{Trajectories on $\beta$--$r$ plane}
        \label{fig:drift_rollouts_beta_r}
    \end{subfigure}
    \hfill
    \begin{subfigure}[t]{0.46\linewidth}
        \centering
        \includegraphics[width=\linewidth]{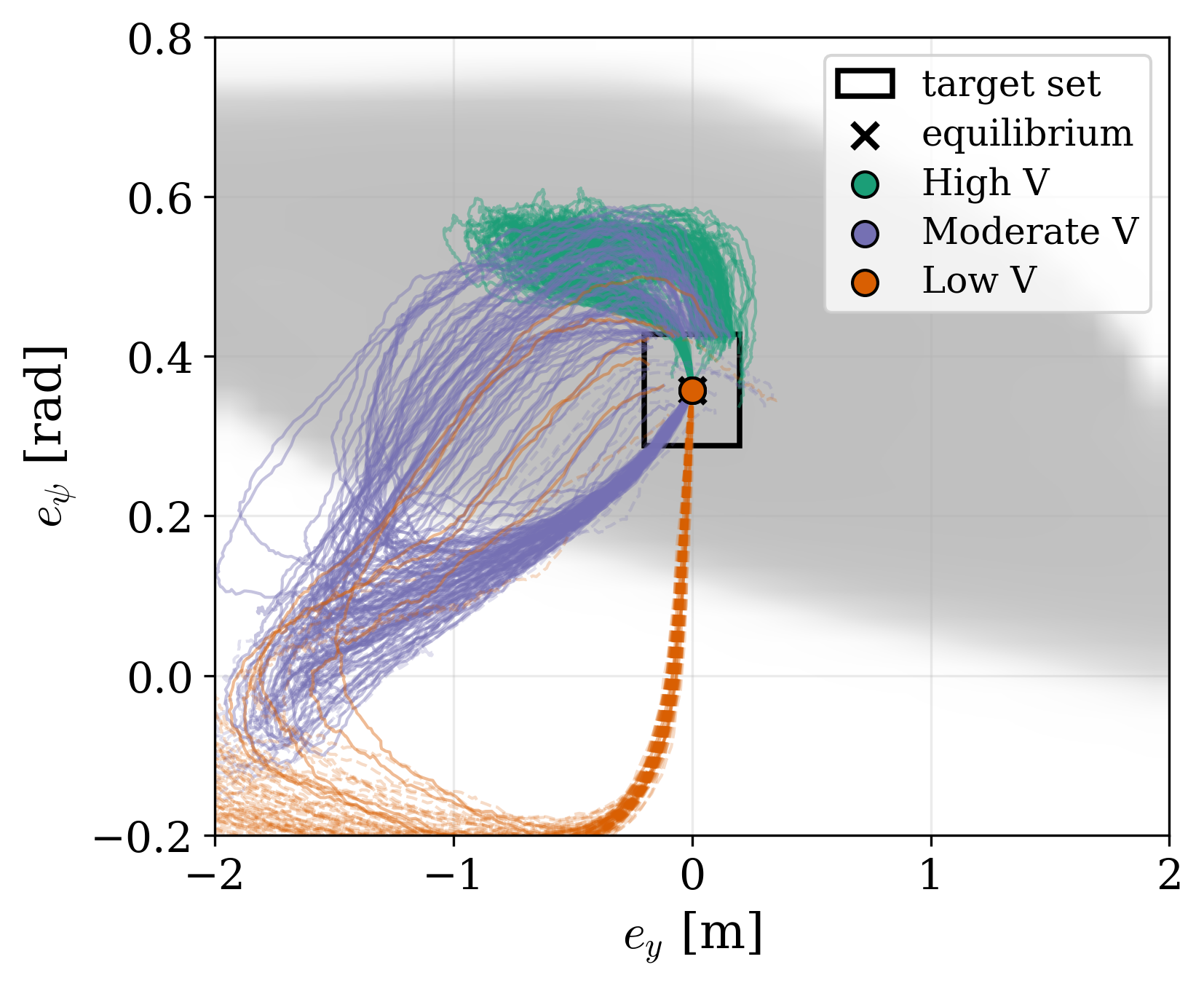}
        \caption{Trajectories on $e_y$--$e_\psi$ plane}
        \label{fig:drift_rollouts_ey_epsi}
    \end{subfigure}
    \caption{Closed-loop rollout results}
    \label{fig:drift_results}
\end{figure}

\section{Conclusions}

This paper proposed a PIRL framework for stochastic reach-avoid set characterization. 
We first established an equivalent RL formulation of the stochastic reach-avoid probability and then developed a scheduled PIRL algorithm that combines TD-based actor-critic learning with HJB residual and boundary-condition losses. 
The learned critic can be related to the HJB residual and boundary mismatch, giving a residual-based interpretation of the value-function error.
Numerical experiments demonstrated the effectiveness of the proposed framework. 
In the one-dimensional reachability problem, PIRL achieved accuracy comparable to successful PINN training while avoiding the failure modes observed in PINN-based policy iteration. 
In the drifting example, PIRL learned a reach-avoid value function consistent with both projected phase portraits and closed-loop rollout outcomes. 
These results indicate that combining trajectory-based RL with physics-informed losses is a promising approach for learning high-dimensional stochastic reach-avoid value functions.

Future work includes extending the proposed framework to more complex high-dimensional safety-critical systems, such as power grids, and using the learned reach-avoid value function for safe controller design.

\appendix

\subsection{Proof of Proposition\,\ref{prop:RL}} \label{proof:RL}

\begin{proof}

To begin the proof, we characterize the event associated with the probability $\Psi^\mu_\mathrm{RA}(\tau, x)$ using the augmented system~\eqref{eq:augmented_dynamics}. Let us define the following event $\mathcal{E}_1$:
\begin{align}
 \mathcal{E}_1 := \bigcup_{j=0}^{N(\tau)} \left( \{ \tilde{X}_j \in A \} \cap \bigcap_{k=0}^{j-1} \{ \tilde{X}_k \notin A \cup B \} \right). 
\end{align}
The event $\mathcal{E}_1$ directly represents the probability $\Psi^\mu_\mathrm{RA}(\tau, x)$ defined in \eqref{eq:reach_avoid_prob}, since it describes the scenario in which the process enters the target set $A$ for the first time at some time $j \in \mathcal{N}_\tau$, while avoiding the unsafe set $B$ at all previous steps $k < j$.
Note that, by the absorbing structure introduced in \eqref{eq:absorbing_structure}, the state $\tilde{X}_k$ evolves identically to the original system state $X_k$ as long as the trajectory remains outside $A \cup B$, and it follows that $\tilde{X}_k = X_k$ for all $k \leq j$ on $\mathcal{E}_1$.
Therefore, 
\begin{align}
    \Psi^\mu_\mathrm{RA}(\tau, x) = \mathbb{P}(\mathcal{E}_1 \,|\, Y_0 = y, \mu).
\end{align}

Next, we show that the event $\mathcal{E}_1$ is equivalent to the following event $\mathcal{E}_2$:
\begin{align}
\mathcal{E}_2 := { \tilde{X}_{N(\tau)} \in A }.
\end{align}
First, suppose that $\omega \in \mathcal{E}_1$. Then, by definition, there exists some $j \in \mathcal{N}_\tau$ such that $\tilde{X}_j \in A$ and $\tilde{X}_k \notin A \cup B$ for all $k < j$.
By the absorbing property of the system, the state remains in $A$ for all $k > j$. 
Thus, $\tilde{X}_{N(\tau)} \in A$, and hence $\omega \in \mathcal{E}_2$.
Conversely, suppose that $\omega \in \mathcal{E}_2$, \ie, $\tilde{X}_{N(\tau)} \in A$.
Then, by the absorbing property, there must exist some $j \le N(\tau)$ such that $\tilde{X}_j \in A$ and $\tilde{X}_k \notin A \cup B$ for all $k < j$. Indeed, if the trajectory had ever entered $B$, then the state would be absorbed in $B$ and could not reach $A$ thereafter; and if it had entered $A$ earlier, the state would have remained in $A$ up to time $N(\tau)$. Let $j^\ast$ be the smallest index such that $\tilde{X}_{j^\ast} \in A$. Then $\tilde{X}_k \notin A \cup B$ for all $k < j^\ast$, and hence $\omega \in \mathcal{E}_1$.
Therefore, $\mathcal{E}_1 = \mathcal{E}_2$, and we have
\begin{align} \label{eq:proof_1_A}
    \Psi^\mu_\mathrm{RA}(\tau, x) = \mathbb{P}(\tilde{X}_{N(\tau)} \in A \mid Y_0 = y, \mu).
\end{align}

Now, consider the value function $v^\mu_\mathrm{RL}(y)$ defined in~\eqref{eq:rl_value_function}. 
Since $T_k = \tau - k \Delta t$, we have $T_k \in \mathcal{T}$ if and only if $k = N(\tau)$. Therefore, only one term in the summation is non-zero:
\begin{align} \label{eq:proof_1_B}
    v^\mu_\mathrm{RL}(y) = \mathbb{E} \left[ \mI_A(\tilde{X}_{N(\tau)}) \,\middle|\, Y_0 = y, \mu \right].
\end{align}
Therefore, from \eqref{eq:proof_1_A} and \eqref{eq:proof_1_B}, we conclude
\begin{align}
    v^\mu_\mathrm{RL}(y) = \Psi^\mu_\mathrm{RA}(\tau, x).
\end{align}
Since $r(Y_k) \in \{0,1\}$, the value function $v^\mu_\mathrm{RL}(y)$ lies in $[0,1]$.

\end{proof}

\subsection{Proof of Theorem\,\ref{thm:pde_characterization}}  
\label{proof:pde_characterization}

\begin{proof}
We first prove statement (i).
Fix $\mu\in\mathcal{U}_{\mathrm M}$ and $\Delta t>0$.
Since the reward $r_\epsilon(Y_k)$ is nonzero only when
$T_k\in[0,\Delta t)$, we have
\begin{align}
v_{\mathrm{RL},\epsilon}^{\mu}
(y)
=
\mathbb E
\left[
l_\epsilon
\left(
\widetilde X_{N(\tau)}
\right)
\,\middle|\,
Y_0=y,\mu
\right].
\end{align}
By construction, $\{\widetilde X_k\}_{k\ge0}$ is 
a time-discretized approximation of the controlled SDE under the piecewise-constant feedback introduced in Sec.~III, with the sets $A_\epsilon$ and $B$ treated as absorbing sets.
Let $\{\hat X_s\}_{s\ge0}$ be the corresponding continuous-time process. 
Standard time-discretization results for stopped diffusions imply that the approximation error vanishes as $\Delta t\to0$.
Thus, for any compact set $K\subset[0,T]\times O_\epsilon$,
\begin{align}
\lim_{\Delta t\to0}
\sup_{(\tau,x)\in K}
\mathbb E
\left[
\left|
\widetilde X_{N(\tau)}
-
\hat X_\tau
\right|
\right]
=0,
\end{align}
where we used $N(\tau)\Delta t\to\tau$.
Since $l_\epsilon$ is Lipschitz continuous, this implies
\begin{align}
&
\sup_{(\tau,x)\in K}
\left|
v_{\mathrm{RL},\epsilon}^\mu (y)
-
\mathbb E
\left[
l_\epsilon(\hat X_\tau)
\,\middle|\,
\hat X_0=x,\mu
\right]
\right|
\\
&\quad
\le
L_\epsilon
\sup_{(\tau,x)\in K}
\mathbb E
\left[
\left|
\widetilde X_{N(\tau)}-\hat X_{\tau}
\right|
\right]
\to0,
\end{align}
where $L_\epsilon$ is a Lipschitz constant of $l_\epsilon$.
Therefore, the limit
\begin{align}
v_\epsilon^\mu(\tau,x)
:=
\mathbb E
\left[
l_\epsilon(\hat X_\tau)
\,\middle|\,
\hat X_0=x,\mu
\right]
\end{align}
exists uniformly on compact subsets of
$[0,T]\times O_\epsilon$.
Finally, by the stochastic representation of solutions to the
Kolmogorov backward equation for the stopped diffusion,
$v_\epsilon^\mu$ is the unique classical solution of
the linear PDE~\eqref{eq:linear_PDE} with the boundary
condition~\eqref{eq:parabolic_boundary}.

We next prove statement (ii).
First, since every piecewise-constant feedback control used
in the discrete-time RL formulation is an admissible control
for the continuous-time diffusion, we have
\begin{align}
v_{\mathrm{RL},\epsilon}^{*}(y)
\le
v_\epsilon^*(N(\tau)\Delta t,x).
\end{align}
Here, the right-hand side is the optimal continuous-time value
with remaining horizon $N(\tau)\Delta t$.
Since $N(\tau)\Delta t\to\tau$ as $\Delta t\to0$ and
$v_\epsilon^*$ is continuous, it follows that, for every
compact set $K\subset[0,T]\times O_\epsilon$,
\begin{align}
\limsup_{\Delta t\to0}
\sup_{(\tau,x)\in K}
\left(
v_{\mathrm{RL},\epsilon}^{*}(y)
-
v_\epsilon^*(\tau,x)
\right)
\le 0.
\label{eq:inequality_1}
\end{align}
Conversely, let $\mu^*\in\mathcal U_{\mathrm M}$ be an
optimal Markov policy for the continuous-time smoothed
reach-avoid problem. Such a policy exists under the classical
solvability assumptions stated in \cref{sec:pde}, and it satisfies
$v_\epsilon^{\mu^*}=v_\epsilon^*$.
On the other hand, by the definition of the discrete-time optimal RL value function,
\begin{align}
v_{\mathrm{RL},\epsilon}^{*}(y)
\ge
v_{\mathrm{RL},\epsilon}^{\mu^*}(y).
\end{align}
By statement (i), $v_{\mathrm{RL},\epsilon}^{\mu^*}$ converges
uniformly on compact subsets of $[0,T]\times O_\epsilon$ to
$v_\epsilon^{\mu^*}=v_\epsilon^*$. Hence
\begin{align}
\liminf_{\Delta t\to0}
\inf_{(\tau,x)\in K}
\left(
v_{\mathrm{RL},\epsilon}^{*}(y)
-
v_\epsilon^*(\tau,x)
\right)
\ge 0.
\label{eq:inequality_2}
\end{align}
Combining inequalities \eqref{eq:inequality_1} and \eqref{eq:inequality_2} yields
\begin{align}
\lim_{\Delta t\to0}
\sup_{(\tau,x)\in K}
\left|
v_{\mathrm{RL},\epsilon}^{*}(y)
-
v_\epsilon^*(\tau,x)
\right|
=0.
\end{align}
Finally, by the dynamic programming principle and the
verification theorem for the smoothed continuous-time
reach-avoid problem, $v_\epsilon^*$ is the unique classical
solution of \eqref{eq:value_PDE} with the boundary condition~\eqref{eq:parabolic_boundary}.
\end{proof}

\subsection{Proof of Corollary\,\ref{thm:error_bound}}  
\label{proof:error_bound}

\begin{proof}
The proof follows the same argument as the error
analysis in~\cite{Wang2026}.
Specifically, the comparison principle for uniformly parabolic
equations is first applied to relate the approximation error to
the boundary mismatch and the PDE residual associated with the
learned policy.
The desired upper bound \eqref{eq:error_upper_bound} is then obtained by noting that the
value function under any admissible policy does not exceed the
optimal reach-avoid value function.
If $\mu_\phi=\mu^*$, the policy evaluation equation coincides
with the optimal HJB equation, and the corresponding PINN
approximation estimate gives the two-sided error
bound \eqref{eq:two_side_bound}.
\end{proof}

\subsection{Prompt for LLM-Guided PIRL Scheduling}
\label{app:codex_prompt}

The prompt consisted of three components: a fixed scheduling template, a user-provided context block containing phase-specific guidance and observations from previous trials, and an automatically generated JSON summary of the completed runs.
The coding agent was required to return only a JSON scheduling plan, which was parsed and executed by the orchestration script.
The prompt template was as follows. 

\begin{Verbatim}[
fontsize=\scriptsize,
breaklines=true,
breakanywhere=true,
breaksymbolleft={},
breaksymbolright={},
frame=single,
framesep=2mm,
xleftmargin=1mm,
xrightmargin=1mm,
baselinestretch=0.90
]
You are controlling the next round of PIRL weight and
collocation-distribution scheduling.

Objective:
- Treat {target_total_updates} total updates as the first milestone, not a hard stop.
- By that milestone, outperform the TD3 baseline from  {baseline_checkpoint}.
- If reward and MC reachability remain stable, keep progressing beyond the milestone.
- Keep final reward no worse than TD3 while reducing value calibration error mean|MC-V|.

Output:
- Write ONLY valid JSON to: {next_plan_path}
- Return exactly {max_parallel_candidates} candidate(s).

Schema:
{ "round_note": "brief rationale",
  "candidates": [
    { "name": "short_unique_name",
      "start_checkpoint": "path/to/ckpt-N",
      "schedule_initial": [1.0, hjb0, bdr0],
      "schedule_final": [1.0, hjb1, bdr1],
      "schedule_center": 500000,
      "schedule_sharpness": 1e-5, 
      }
  ]
}

Selection rules:
- Continue from the best safe checkpoint when reward and MC are stable.
- If reward or meanMC degraded, reduce weights or slow the schedule before trying larger weights.
- Increase HJB/BDR gradually.
- Increase replay_expand jitter gradually; prefer holding or backing off jitter before increasing HJB/BDR when reward or meanMC weakens.
- Candidate-level pinn_* fields are optional; omitted fields inherit [training_env] defaults from the TOML config.
- Use at most one TD3-restart control per round, unless all scheduling checkpoints collapsed.
- Do not repeat an existing start_checkpoint + schedule_initial + schedule_final + pinn_expand_jitter_final combination unless round_note explains why.
{context_block}

Completed results JSON:
{completed_results_json}
\end{Verbatim}

\section*{Acknowledgment}

The authors would like to thank Naoki Hashima and Eiko Furutani for discussions on the interpretation of the RL characterization of the reach-avoid problem, and Yuxuan Zhu for reporting the PINN failure mode in PIRL experiments.

\bibliographystyle{IEEEtran}
\bibliography{reference}

@inproceedings{omura2025gradual,
  title     = {Gradual Transition from Bellman Optimality Operator to Bellman Operator in Online Reinforcement Learning},
  author    = {Omura, Motoki and Ota, Kazuki and Osa, Takayuki and Mukuta, Yusuke and Harada, Tatsuya},
  booktitle = {Proceedings of the 42nd International Conference on Machine Learning},
  series    = {Proceedings of Machine Learning Research},
  volume    = {267},
  year      = {2025}
}

@inproceedings{asadi2017alternative,
title={An Alternative to Softmax in Reinforcement Learning},
author={Asadi, Kavosh and Littman, Michael L.},
booktitle={Proceedings of the 34th International Conference on Machine Learning (ICML)},
series    = {Proceedings of Machine Learning Research},
pages={243--252},
year={2017}
}

@inproceedings{kostrikov2022iql,
title={Offline Reinforcement Learning with Implicit {Q}-Learning},
author={Kostrikov, Ilya and Nair, Ashvin and Levine, Sergey},
booktitle={International Conference on Learning Representations (ICLR)},
series    = {Proceedings of Machine Learning Research},
year={2022}
}

@ARTICLE{Wang2026,
  author={Wang, Zhuoyuan and Chern, Albert and Nakahira, Yorie},
  journal={IEEE Transactions on Automatic Control}, 
  title={Generalizable Physics-Informed Learning for Stochastic Safety-Critical Systems}, 
  year={2026},
  volume={71},
  number={4},
  pages={2155-2170},
  doi={10.1109/TAC.2025.3622907}}

@article{Xue2026,
title = {Sufficient and necessary barrier-like conditions for safety and reach-avoid verification of stochastic discrete-time systems},
journal = {Automatica},
volume = {187},
pages = {112919},
year = {2026},
doi = {https://doi.org/10.1016/j.automatica.2026.112919},
author = {Bai Xue},
}

@article{Luo2025,
  title={Physics-informed neural networks for PDE problems: a comprehensive review},
  author={Luo, Kuang and Zhao, Jingshang and Wang, Yingping and Li, Jiayao and Wen, Junjie and Liang, Jiong and Soekmadji, Henry and Liao, Shaolin},
  journal={Artificial Intelligence Review},
  volume={58},
  number={10},
  pages={323},
  year={2025},
  publisher={Springer}
}

@article{Banerjee2025,
title = {A survey on physics informed reinforcement learning: Review and open problems},
journal = {Expert Systems with Applications},
volume = {287},
pages = {128166},
year = {2025},
doi = {https://doi.org/10.1016/j.eswa.2025.128166},
author = {Chayan Banerjee and Kien Nguyen and Clinton Fookes and Maziar Raissi},
}

@ARTICLE{Ganai2024,
  author={Ganai, Milan and Gao, Sicun and Herbert, Sylvia L.},
  journal={IEEE Open Journal of Control Systems}, 
  title={{Hamilton-Jacobi} Reachability in Reinforcement Learning: A Survey}, 
  year={2024},
  volume={3},
  number={},
  pages={310-324},
  doi={10.1109/OJCSYS.2024.3449138}}

@inproceedings{Meng2024,
author = {Meng, Yiming and Zhou, Ruikun and Mukherjee, Amartya and Fitzsimmons, Maxwell and Song, Christopher and Liu, Jun},
title = {Physics-informed neural network policy iteration: algorithms, convergence, and verification},
year = {2024},
booktitle = {Proceedings of the 41st International Conference on Machine Learning},
articleno = {1441},
numpages = {26},
location = {Vienna, Austria},
series_ = {ICML'24}
}

@article{Wang2024PIRL,
author = {Wang, Yujia and Wu, Zhe},
title = {Physics-informed reinforcement learning for optimal control of nonlinear systems},
journal = {AIChE Journal},
volume = {70},
number = {10},
pages = {e18542},
doi = {https://doi.org/10.1002/aic.18542},
year = {2024}
}

@INPROCEEDINGS{HoshinoACC2024,
  author={Hoshino, Hikaru and Nakahira, Yorie},
  booktitle={2024 American Control Conference (ACC)}, 
  title={Physics-informed {RL} for Maximal Safety Probability Estimation}, 
  year={2024},
  pages={3576--3583},
  doi={10.23919/ACC60939.2024.10644621}}

@INPROCEEDINGS{HoshinoITSC2024,
  author={Hoshino, Hikaru and Li, Jiaxing and Menon, Arnav and Dolan, John M. and Nakahira, Yorie},
  booktitle={2024 IEEE 27th International Conference on Intelligent Transportation Systems (ITSC)}, 
  title={Autonomous Drifting Based on Maximal Safety Probability Learning}, 
  year={2024},
  pages={3930--3935},
  doi={10.1109/ITSC58415.2024.10919509}}

@inproceedings{Zikelic2023,
  title={Learning control policies for stochastic systems with reach-avoid guarantees},
  author={{\v{Z}}ikeli{\'c}, {\DJ}or{\dj}e and Lechner, Mathias and Henzinger, Thomas A and Chatterjee, Krishnendu},
  booktitle={Proceedings of the AAAI Conference on Artificial Intelligence},
  volume={37},
  number={10},
  pages={11926--11935},
  year={2023}
}

@inproceedings{Mukherjee2023,
    title={Bridging Physics-Informed Neural Networks with Reinforcement Learning: {Hamilton-Jacobi-Bellman} Proximal Policy Optimization ({HJBPPO})},
    author={Amartya Mukherjee and Jun Liu},
    booktitle={ICML Workshop on New Frontiers in Learning, Control, and Dynamical Systems},
    year={2023},
}

@article{S.Wang2023:PinnGuide,
  title={An expert's guide to training physics-informed neural networks},
  author={Wang, Sifan and Sankaran, Shyam and Wang, Hanwen and Perdikaris, Paris},
  journal={arXiv preprint arXiv:2308.08468},
  year={2023}
}

@inproceedings{Schmid2023,
title = {Probabilistic Reachability and Invariance Computation of Stochastic Systems using Linear Programming},
pages = {11229-11234},
year = {2023},
booktitle = {22nd IFAC World Congress},
doi = {https://doi.org/10.1016/j.ifacol.2023.10.853},
author = {Niklas Schmid and John Lygeros},
}

@article{Liao2022,
author = {Wei Liao and Taotao Liang and Xiaohui Wei and Jizhou Lai and},
title = {A novel unified framework for solving reachability and invariance problems},
journal = {International Journal of Control},
volume = {96},
number = {6},
pages = {1436--1447},
year = {2023},
publisher = {Taylor \& Francis},
doi = {10.1080/00207179.2022.2051749},
}

@article{Thorpe2022,
title = {State-based confidence bounds for data-driven stochastic reachability using Hilbert space embeddings},
journal = {Automatica},
volume = {138},
pages = {110146},
year = {2022},
doi = {https://doi.org/10.1016/j.automatica.2021.110146},
author = {Adam J. Thorpe and Kendric R. Ortiz and Meeko M.K. Oishi},
}

@article{Leiteritz2021,
  title={How to avoid trivial solutions in physics-informed neural networks},
  author={Leiteritz, Raphael and Pfl{\"u}ger, Dirk},
  journal={arXiv preprint arXiv:2112.05620},
  year={2021}
}

@ARTICLE{Hsu2021,
  title         = "Safety and Liveness Guarantees through {Reach-Avoid} Reinforcement Learning",
  author        = "Hsu, Kai-Chieh and Rubies-Royo, Vicen{\c c} and Tomlin, Claire J and Fisac, Jaime F",
  month         =  dec,
  year          =  2021,
  archivePrefix = "arXiv",
  primaryClass  = "cs.LG",
  eprint        = "2112.12288",
  journal = {arXiv: 2112.12288 [cs.LG]}
}

@INPROCEEDINGS{Xue2021,
  author={Xue, Bai and Li, Renjue and Zhan, Naijun and Fränzle, Martin},
  booktitle={2021 American Control Conference (ACC)}, 
  title={Reach-avoid Analysis for Stochastic Discrete-time Systems}, 
  year={2021},
  pages={4879-4885},
  doi={10.23919/ACC50511.2021.9483095}}

@INPROCEEDINGS{Chern21,
  author={Chern, Albert and Wang, Xiang and Iyer, Abhiram and Nakahira, Yorie},
  booktitle={2021 60th IEEE Conference on Decision and Control (CDC)}, 
  title={Safe Control in the Presence of Stochastic Uncertainties}, 
  year={2021},
  volume={},
  number={},
  pages={6640-6645},
  doi={10.1109/CDC45484.2021.9683542}}

@ARTICLE{Thorpe2020,
  author={Thorpe, Adam J. and Oishi, Meeko M. K.},
  journal={IEEE Control Systems Letters}, 
  title={Model-Free Stochastic Reachability Using Kernel Distribution Embeddings}, 
  year={2020},
  volume={4},
  number={2},
  pages={512--517},
  doi={10.1109/LCSYS.2019.2954102}}

@INPROCEEDINGS{Fisac2019,
  title     = "Bridging {Hamilton-Jacobi} Safety Analysis and Reinforcement Learning",
  booktitle = "2019 International Conference on Robotics and Automation ({ICRA})",
  author    = "Fisac, Jaime F and Lugovoy, Neil F and Rubies-Royo, Vicen{\c c} and Ghosh, Shromona and Tomlin, Claire J",
  pages     = "8550--8556",
  month     =  may,
  year      =  2019,
}

@ARTICLE{Raissi2019,
  title    = "Physics-informed neural networks: A deep learning framework for solving forward and inverse problems involving nonlinear partial differential equations",
  author   = "Raissi, M and Perdikaris, P and Karniadakis, G E",
  journal  = "J. Comput. Phys.",
  volume   =  378,
  pages    = "686--707",
  month    =  feb,
  year     =  2019,
}

@article{Bobier2019,
  title={Vehicle control synthesis using phase portraits of planar dynamics},
  author={Bobier-Tiu, Carrie G and Beal, Craig E and Kegelman, John C and Hindiyeh, Rami Y and Gerdes, J Christian},
  journal={Vehicle System Dynamics},
  volume={57},
  number={9},
  pages={1318--1337},
  year={2019},
  publisher={Taylor \& Francis}
}

@article{Akametalu2018,
  title={A minimum discounted reward hamilton-jacobi formulation for computing reachable sets},
  author={Akametalu, Anayo K and Ghosh, Shromona and Fisac, Jaime F and Tomlin, Claire J},
  journal={arXiv preprint arXiv:1809.00706},
  year={2018}
}

@inproceedings{Fujimoto2018:TD3,
  title={Addressing function approximation error in actor-critic methods},
  author={Fujimoto, Scott and Hoof, Herke and Meger, David},
  booktitle={International conference on machine learning},
  pages={1587--1596},
  year={2018},
  organization={PMLR}
}

@INPROCEEDINGS{Gleason2017,
  author={Gleason, Joseph D. and Vinod, Abraham P. and Oishi, Meeko M. K.},
  booktitle={2017 IEEE 56th Annual Conference on Decision and Control (CDC)}, 
  title={Underapproximation of reach-avoid sets for discrete-time stochastic systems via Lagrangian methods}, 
  year={2017},
  pages={4283--4290},
  doi={10.1109/CDC.2017.8264291}}

@INPROCEEDINGS{Bansal2017,
  title     = "{Hamilton-Jacobi} reachability: A brief overview and recent advances",
  booktitle = "2017 {IEEE} 56th Annual Conference on Decision and Control ({CDC})",
  author    = "Bansal, Somil and Chen, Mo and Herbert, Sylvia and Tomlin, Claire J",
  pages     = "2242--2253",
  month     =  dec,
  year      =  2017,
}

@article{MohajerinEsfahani2016,
title = {The stochastic reach-avoid problem and set characterization for diffusions},
journal = {Automatica},
volume = {70},
pages = {43-56},
year = {2016},
issn = {0005-1098},
doi = {https://doi.org/10.1016/j.automatica.2016.03.016},
author = {Peyman {Mohajerin Esfahani} and Debasish Chatterjee and John Lygeros},
}

@article{Lillicrap2015:DDPG,
  title={Continuous control with deep reinforcement learning},
  author={Lillicrap, Timothy P. and Hunt, Jonathan J. and Pritzel, Alexander and Heess, Nicolas and Erez, Tom and Tassa, Yuval and Silver, David and Wierstra, Daan},
  journal={arXiv preprint arXiv:1509.02971},
  year={2015}
}

@article{Mnih15,
    author = {Volodymyr Mnih and Koray Kavukcuoglu and David Silver and Andrei A. Rusu and Joel Veness and Marc G. Bellemare and  Alex Graves and Martin Riedmiller and Andreas K. Fidjeland and Georg Ostrovski and Stig Petersen and Charles Beattie and Amir Sadik and  Ioannis Antonoglou and Helen King and Dharshan Kumaran and  Daan Wierstra and Shane Legg and Demis Hassabis},
    title = {Human-level control through deep reinforcement learning},
    journal = {Nature},
    year = 2015,
    volume = 518,
    pages = {529--533},
}

@ARTICLE{Hindiyeh2014,
  title     = "A Controller Framework for Autonomous Drifting: Design,
               Stability, and Experimental Validation",
  author    = "Hindiyeh, Rami Y and Christian Gerdes, J",
  journal   = "J. Dyn. Syst. Meas. Control",
  publisher = "American Society of Mechanical Engineers Digital Collection",
  volume    =  136,
  number    =  5,
  pages     = "051015",
  month     =  jul,
  year      =  2014,
}

@article{Mao2013,
title = {Stabilization of continuous-time hybrid stochastic differential equations by discrete-time feedback control},
journal = {Automatica},
volume = {49},
number = {12},
pages = {3677-3681},
year = {2013},
issn = {0005-1098},
doi = {https://doi.org/10.1016/j.automatica.2013.09.005},
author = {Xuerong Mao},
}

@ARTICLE{Summers2010,
  title    = "Verification of discrete time stochastic hybrid systems: A stochastic reach-avoid decision problem",
  author   = "Summers, Sean and Lygeros, John",
  journal  = "Automatica",
  volume   =  46,
  number   =  12,
  pages    = "1951--1961",
  month    =  dec,
  year     =  2010,
}

@ARTICLE{Abate2008,
  title    = "Probabilistic reachability and safety for controlled discrete
              time stochastic hybrid systems",
  author   = "Abate, Alessandro and Prandini, Maria and Lygeros, John and Sastry, Shankar",
  journal  = "Automatica",
  volume   =  44,
  number   =  11,
  pages    = "2724--2734",
  month    =  nov,
  year     =  2008
}

@book{Fleming06,
    author = {W.~H. Fleming and H.~M. Soner},
    title = {Controlled Markov Processes and Viscosity Solutions},
    edition = {2nd},
    publisher = {Springer},
    year = 2006
}

@ARTICLE{Mitchell2005,
  author={Mitchell, I.M. and Bayen, A.M. and Tomlin, C.J.},
  journal={IEEE Transactions on Automatic Control}, 
  title={A time-dependent {Hamilton-Jacobi} formulation of reachable sets for continuous dynamic games}, 
  year={2005},
  volume={50},
  number={7},
  pages={947-957},
  doi={10.1109/TAC.2005.851439}}

@article{Lygeros2004,
  author    = {John Lygeros},
  title     = {On Reachability and Minimum Cost Optimal Control},
  journal   = {Automatica},
  year      = {2004},
  volume    = {40},
  number    = {6},
  pages     = {917--927}
}

@book{Karatzas1998,
  title={Brownian motion and stochastic calculus},
  author={Karatzas, Ioannis and Shreve, Steven},
  edition={2},
  year={1998},
  publisher={Springer}
}



\end{document}